\documentclass{llncs}

\usepackage{multicol}

\usepackage{amsmath}
\usepackage{amsfonts}
\usepackage{amssymb}

\usepackage{graphicx}
\usepackage{epic}
\usepackage{eepic}
\usepackage{epsfig,float}

\usepackage{pdfsync} 

\renewcommand{\le}{\leqslant}
\renewcommand{\ge}{\geqslant}

\newcommand{\eps}{\varepsilon}
\newcommand{\emp}{\emptyset}

\newcommand{\Sig}{\Sigma}
\newcommand{\sig}{\sigma}
\newcommand{\noin}{\noindent}

\newcommand{\ur}{uniquely reachable}
\newcommand{\bi}{\begin{itemize}}
\newcommand{\ei}{\end{itemize}}
\newcommand{\be}{\begin{enumerate}}
\newcommand{\ee}{\end{enumerate}}
\newcommand{\bd}{\begin{description}}
\newcommand{\ed}{\end{description}}
\newcommand{\bq}{\begin{quote}}
\newcommand{\eq}{\end{quote}}

\newcommand{\ie}{\mbox{\it i.e.}}

\newcommand{\cA}{{\mathcal A}}

\newcommand{\cN}{{\mathcal N}}

\newcommand{\cR}{{\mathcal R}}
\newcommand{\cS}{{\mathcal S}}

\newcommand{\Lra}{{\hspace{.1cm}\Leftrightarrow\hspace{.1cm}}}

\newcommand{\raL}{{\hspace{.1cm}{\rightarrow_L} \hspace{.1cm}}}
\newcommand{\lraL}{{\hspace{.1cm}{\leftrightarrow_L} \hspace{.1cm}}}

\title{Syntactic Complexity of Ideal and Closed Languages
\thanks{This work was supported by the Natural Sciences and Engineering Research Council of Canada under grant No.~OGP0000871 and  a Postgraduate Scholarship, and by a Graduate Award from the Department of Computer Science, University of Toronto.
}
}

\author{Janusz~Brzozowski\inst{1} \and Yuli Ye\inst{2}}

\authorrunning{Brzozowski, Ye}   

\institute{David R. Cheriton School of Computer Science, University of Waterloo, \\
Waterloo, ON, Canada N2L 3G1\\
\{{\tt brzozo@uwaterloo.ca} \}
\and
Department of Computer Science, University of Toronto,\\
 Toronto, ON,  Canada M5S 3G4,\\
\{{\tt y3ye@cs.toronto.edu}\}
}

\begin{document}

\maketitle
\today
\begin{abstract}
The state complexity of a regular language is the number of states in the minimal deterministic automaton accepting the language.
The syntactic complexity of a regular language is the cardinality of its syntactic semigroup.
The syntactic complexity of a subclass of regular languages is the worst-case syntactic complexity taken as a function of the state complexity $n$ of languages in that class.
We study the syntactic complexity of the class of regular ideal  languages and their complements, the closed languages.
We prove that $n^{n-1}$ is a tight upper bound on the complexity of right ideals and prefix-closed languages, and that
there exist left ideals and suffix-closed languages of syntactic complexity $n^{n-1}+n-1$, and two-sided ideals and factor-closed languages
of syntactic complexity $n^{n-2}+(n-2)2^{n-2}+1$.
\bigskip

\noin
\end{abstract}

\section{Introduction}

There are two fundamental congruence relations in the theory of regular languages: the Nerode congruence~\cite{Ner58}, and the Myhill congruence~\cite{Myh57}. In both cases, a language is regular if and only if it is a union of congruence classes of a congruence of finite index.
The Nerode congruence leads to the definitions of left quotients of a language and the minimal deterministic finite automaton  recognizing the language. The Myhill congruence leads to the definitions of the syntactic semigroup and the syntactic monoid of the language.

The \emph{state complexity} of a language is defined as the number of states in the minimal deterministic automaton recognizing the language. This concept has been studied quite extensively: for surveys of this topic and lists of references we refer the reader to~\cite{Brz10,Yu01}.
On the other hand, in spite of suggestions that  syntactic semigroups deserve to be studied further~\cite{HoKo04,KLS05}, relatively little has been done on the ``syntactic complexity'' of a regular language, which we define as the cardinality of the syntactic semigroup of the language.
This semigroup is isomorphic to the semigroup of transformations of the set of states of the minimal deterministic automaton recognizing the language, where these transformations are performed  by non-empty words.

The following example illustrates the significant difference between state complexity and syntactic complexity.
\begin{example}
\label{ex:3automata}
The deterministic automata in Fig.~\ref{fig:3automata} have the same alphabet, are all minimal,  and have the same state complexity. However, the syntactic complexity of $\cA_1$ is 3, that of $\cA_2$ is 9, and that of $\cA_3$ is 27. 
\end{example}

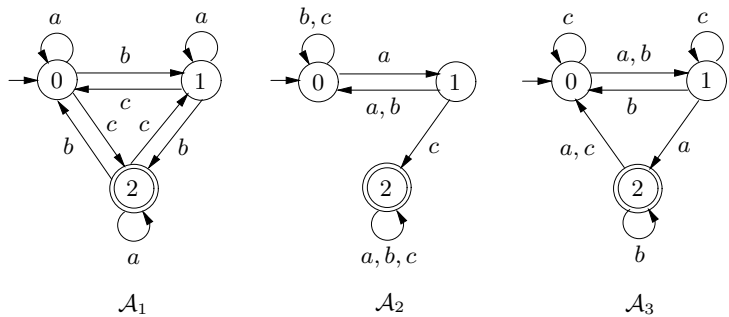
\begin{figure}[hbt]
\begin{center}
\setlength{\unitlength}{0.00056868in}
\begingroup\makeatletter\ifx\SetFigFont\undefined%
\gdef\SetFigFont#1#2#3#4#5{%
  \reset@font\fontsize{#1}{#2pt}%
  \fontfamily{#3}\fontseries{#4}\fontshape{#5}%
  \selectfont}%
\fi\endgroup%
{\renewcommand{\dashlinestretch}{30}
\begin{picture}(6677,2936)(0,-10)
\put(1040,1955){\makebox(0,0)[lb]{\smash{{\SetFigFont{9}{10.8}{\familydefault}{\mddefault}{\updefault}$c$}}}}
\put(1801.000,2526.000){\arc{298.530}{2.1848}{7.2400}}
\blacken\path(1628.118,2492.042)(1715.000,2404.000)(1679.771,2522.570)(1628.118,2492.042)
\put(1171.000,879.480){\arc{299.380}{5.3244}{10.3836}}
\blacken\path(1344.575,914.646)(1257.000,1002.000)(1293.164,883.712)(1344.575,914.646)
\put(3526.000,887.480){\arc{299.380}{5.3244}{10.3836}}
\blacken\path(3699.575,922.646)(3612.000,1010.000)(3648.164,891.712)(3699.575,922.646)
\put(2881.000,2504.000){\arc{298.530}{2.1848}{7.2400}}
\blacken\path(2708.118,2470.042)(2795.000,2382.000)(2759.771,2500.570)(2708.118,2470.042)
\put(5214.000,2511.000){\arc{298.530}{2.1848}{7.2400}}
\blacken\path(5041.118,2477.042)(5128.000,2389.000)(5092.771,2507.570)(5041.118,2477.042)
\put(6489.000,2526.000){\arc{298.530}{2.1848}{7.2400}}
\blacken\path(6316.118,2492.042)(6403.000,2404.000)(6367.771,2522.570)(6316.118,2492.042)
\put(5859.000,902.480){\arc{299.380}{5.3244}{10.3836}}
\blacken\path(6032.575,937.646)(5945.000,1025.000)(5981.164,906.712)(6032.575,937.646)
\put(462,2228){\ellipse{368}{368}}
\put(2884,2209){\ellipse{368}{368}}
\put(4160,2212){\ellipse{368}{368}}
\put(3522,1227){\ellipse{368}{368}}
\put(3520,1230){\ellipse{450}{450}}
\put(1798,2222){\ellipse{368}{368}}
\put(5232,2217){\ellipse{368}{368}}
\put(6485,2221){\ellipse{368}{368}}
\put(5848,1221){\ellipse{368}{368}}
\put(5847,1219){\ellipse{450}{450}}
\put(1174,1222){\ellipse{368}{368}}
\put(1175,1219){\ellipse{450}{450}}
\path(12,2224)(275,2224)
\blacken\path(155.000,2194.000)(275.000,2224.000)(155.000,2254.000)(155.000,2194.000)
\path(1145,1452)(1670,2097)
\blacken\path(1617.514,1984.994)(1670.000,2097.000)(1570.981,2022.871)(1617.514,1984.994)
\path(610,2103)(1083,1428)
\blacken\path(989.567,1509.058)(1083.000,1428.000)(1038.704,1543.490)(989.567,1509.058)
\path(2441,2224)(2704,2224)
\blacken\path(2584.000,2194.000)(2704.000,2224.000)(2584.000,2254.000)(2584.000,2194.000)
\path(4775,2217)(5038,2217)
\blacken\path(4918.000,2187.000)(5038.000,2217.000)(4918.000,2247.000)(4918.000,2187.000)
\path(4100,2044)(3650,1414)
\blacken\path(3695.337,1529.085)(3650.000,1414.000)(3744.161,1494.211)(3695.337,1529.085)
\path(6408,2035)(5958,1405)
\blacken\path(6003.337,1520.085)(5958.000,1405.000)(6052.161,1485.211)(6003.337,1520.085)
\path(5720,1414)(5270,2044)
\blacken\path(5364.161,1963.789)(5270.000,2044.000)(5315.337,1928.915)(5364.161,1963.789)
\path(1805,2044)(1310,1414)
\blacken\path(1360.549,1526.893)(1310.000,1414.000)(1407.728,1489.824)(1360.549,1526.893)
\path(4010,2134)(3065,2134)
\blacken\path(3185.000,2164.000)(3065.000,2134.000)(3185.000,2104.000)(3185.000,2164.000)
\path(965,1306)(470,2026)
\blacken\path(562.705,1944.111)(470.000,2026.000)(513.262,1910.119)(562.705,1944.111)
\path(642,2292)(1632,2292)
\blacken\path(1512.000,2262.000)(1632.000,2292.000)(1512.000,2322.000)(1512.000,2262.000)
\path(1610,2142)(627,2142)
\blacken\path(747.000,2172.000)(627.000,2142.000)(747.000,2112.000)(747.000,2172.000)
\path(3080,2284)(4002,2284)
\blacken\path(3882.000,2254.000)(4002.000,2284.000)(3882.000,2314.000)(3882.000,2254.000)
\path(5405,2292)(6305,2292)
\blacken\path(6185.000,2262.000)(6305.000,2292.000)(6185.000,2322.000)(6185.000,2262.000)
\path(6305,2134)(5390,2134)
\blacken\path(5510.000,2164.000)(5390.000,2134.000)(5510.000,2104.000)(5510.000,2164.000)
\put(1219,1752){\makebox(0,0)[lb]{\smash{{\SetFigFont{9}{10.8}{\familydefault}{\mddefault}{\updefault}$c$}}}}
\put(1099,1152){\makebox(0,0)[lb]{\smash{{\SetFigFont{9}{10.8}{\familydefault}{\mddefault}{\updefault}$2$}}}}
\put(1740,2750){\makebox(0,0)[lb]{\smash{{\SetFigFont{9}{10.8}{\familydefault}{\mddefault}{\updefault}$a$}}}}
\put(382,2748){\makebox(0,0)[lb]{\smash{{\SetFigFont{9}{10.8}{\familydefault}{\mddefault}{\updefault}$a$}}}}
\put(913,1752){\makebox(0,0)[lb]{\smash{{\SetFigFont{9}{10.8}{\familydefault}{\mddefault}{\updefault}$c$}}}}
\put(4094,2134){\makebox(0,0)[lb]{\smash{{\SetFigFont{9}{10.8}{\familydefault}{\mddefault}{\updefault}$1$}}}}
\put(3454,1152){\makebox(0,0)[lb]{\smash{{\SetFigFont{9}{10.8}{\familydefault}{\mddefault}{\updefault}$2$}}}}
\put(403,2149){\makebox(0,0)[lb]{\smash{{\SetFigFont{9}{10.8}{\familydefault}{\mddefault}{\updefault}$0$}}}}
\put(1731,2141){\makebox(0,0)[lb]{\smash{{\SetFigFont{9}{10.8}{\familydefault}{\mddefault}{\updefault}$1$}}}}
\put(2825,2134){\makebox(0,0)[lb]{\smash{{\SetFigFont{9}{10.8}{\familydefault}{\mddefault}{\updefault}$0$}}}}
\put(2707,2749){\makebox(0,0)[lb]{\smash{{\SetFigFont{9}{10.8}{\familydefault}{\mddefault}{\updefault}$b,c$}}}}
\put(5172,2141){\makebox(0,0)[lb]{\smash{{\SetFigFont{9}{10.8}{\familydefault}{\mddefault}{\updefault}$0$}}}}
\put(6426,2149){\makebox(0,0)[lb]{\smash{{\SetFigFont{9}{10.8}{\familydefault}{\mddefault}{\updefault}$1$}}}}
\put(5787,1152){\makebox(0,0)[lb]{\smash{{\SetFigFont{9}{10.8}{\familydefault}{\mddefault}{\updefault}$2$}}}}
\put(5150,2742){\makebox(0,0)[lb]{\smash{{\SetFigFont{9}{10.8}{\familydefault}{\mddefault}{\updefault}$c$}}}}
\put(6395,2750){\makebox(0,0)[lb]{\smash{{\SetFigFont{9}{10.8}{\familydefault}{\mddefault}{\updefault}$c$}}}}
\put(5647,2404){\makebox(0,0)[lb]{\smash{{\SetFigFont{9}{10.8}{\familydefault}{\mddefault}{\updefault}$a,b$}}}}
\put(3435,2391){\makebox(0,0)[lb]{\smash{{\SetFigFont{9}{10.8}{\familydefault}{\mddefault}{\updefault}$a$}}}}
\put(1040,2388){\makebox(0,0)[lb]{\smash{{\SetFigFont{9}{10.8}{\familydefault}{\mddefault}{\updefault}$b$}}}}
\put(1108,521){\makebox(0,0)[lb]{\smash{{\SetFigFont{9}{10.8}{\familydefault}{\mddefault}{\updefault}$a$}}}}
\put(5815,537){\makebox(0,0)[lb]{\smash{{\SetFigFont{9}{10.8}{\familydefault}{\mddefault}{\updefault}$b$}}}}
\put(3277,528){\makebox(0,0)[lb]{\smash{{\SetFigFont{9}{10.8}{\familydefault}{\mddefault}{\updefault}$a,b,c$}}}}
\put(3410,93){\makebox(0,0)[lb]{\smash{{\SetFigFont{9}{10.8}{\familydefault}{\mddefault}{\updefault}$\cA_2$}}}}
\put(5727,78){\makebox(0,0)[lb]{\smash{{\SetFigFont{9}{10.8}{\familydefault}{\mddefault}{\updefault}$\cA_3$}}}}
\put(1025,78){\makebox(0,0)[lb]{\smash{{\SetFigFont{9}{10.8}{\familydefault}{\mddefault}{\updefault}$\cA_1$}}}}
\put(1579,1527){\makebox(0,0)[lb]{\smash{{\SetFigFont{9}{10.8}{\familydefault}{\mddefault}{\updefault}$b$}}}}
\put(515,1541){\makebox(0,0)[lb]{\smash{{\SetFigFont{9}{10.8}{\familydefault}{\mddefault}{\updefault}$b$}}}}
\put(3906,1549){\makebox(0,0)[lb]{\smash{{\SetFigFont{9}{10.8}{\familydefault}{\mddefault}{\updefault}$c$}}}}
\put(6218,1555){\makebox(0,0)[lb]{\smash{{\SetFigFont{9}{10.8}{\familydefault}{\mddefault}{\updefault}$a$}}}}
\put(5123,1533){\makebox(0,0)[lb]{\smash{{\SetFigFont{9}{10.8}{\familydefault}{\mddefault}{\updefault}$a,c$}}}}
\put(5741,1910){\makebox(0,0)[lb]{\smash{{\SetFigFont{9}{10.8}{\familydefault}{\mddefault}{\updefault}$b$}}}}
\put(3323,1908){\makebox(0,0)[lb]{\smash{{\SetFigFont{9}{10.8}{\familydefault}{\mddefault}{\updefault}$a,b$}}}}
\put(459.000,2511.000){\arc{298.530}{2.1848}{7.2400}}
\blacken\path(286.118,2477.042)(373.000,2389.000)(337.771,2507.570)(286.118,2477.042)
\end{picture}
}
\end{center}
\caption{Automata with various syntactic complexities.} 
\label{fig:3automata}
\end{figure}

Syntactic complexity  provides an alternative measure for the complexity of a regular language. The following  question then arises:
\begin{quote}
Is it possible to find upper bounds to the syntactic complexity of a regular language from its properties or from the properties of its minimal deterministic automaton?
\end{quote}
We shed some light on this question for ideal and closed regular languages.

\section{Background}

This section provides a brief informal overview of the past work related to the topic of this paper. The relevant concepts will be formally defined later.

In 1938, Piccard~\cite{Pic38} proved that two generators are sufficient to generate the set of all permutations of a set of $n$ elements, that is, the symmetric group of degree $n$. The two generators can be a cyclic permutation of all the elements and a transposition of two  of the elements.
References to her other early papers can be found in her books~\cite{Pic46,Pic57}, published in 1946, and 1957, where the problem of generators of groups is treated in detail.

In 1960, 1962, and 1963 Salomaa~\cite{Sal60,Sal62,Sal63} studied, among other problems, the sets  that can generate the set of all transformations of a set of $n$ elements.  
In particular, his aim was to replace the symmetric group of degree $n$ by smaller groups of degree $n$. 

In 1968, D\'enes~\cite{Den68} proved that three transformations are sufficient to generate the set of all transformations of a set $S_n$ of $n$ elements. One can use the two transformations that generate the symmetric group of degree $n$ and an additional transformation that maps each of $n-1$ elements of a subset $S_{n-1}$ of $S_n$ to itself, and the last element to some element of $S_{n-1}$. 
Moreover, he showed that fewer than three generators are not possible.
A summary of other work by D\'enes on transformations can be found in~\cite{Den72}.

In 1970, Maslov~\cite{Mas70} dealt with the problem of generators of the semigroup of all transformations in the setting of finite automata. He pointed out that a certain ternary automaton with $n$ states has $n^n$ transformations. He also stated without proof that it is not possible to reach this bound with a binary automaton, and that the precise bound for the binary case is not known.
He exhibited a binary automaton with $n$ states that has at least $(n-1)^{n-1}$ transformations.

In 2002--2004, Holzer and K\"onig~\cite{HoKo02,HoKo03,HoKo04} studied the syntactic complexity of automata.
They remarked that the syntactic complexity of a unary regular language of state complexity $n$  is at most $n$, and this bound can be met. They also noted that $n^n$ is a tight bound on the complexity of languages over alphabets $\Sig$ with $|\Sig|\ge 3$. Their main contributions are in the most difficult case, that of a binary alphabet. They proved that, for $n\ge 3$, the function $n^n-n!+g(n)$ is an upper bound to the syntactic complexity of a binary regular language, where $g(n)$ is the Landau function.
They also showed that a syntactic complexity of $n^n(1-2/\sqrt{n})$ can be achieved. For any prime $n\ge 7$, they characterized a 2-generator semigroup of maximal complexity.

In 2003, Salomaa~\cite{Sal03} considered  all the words over the alphabet $\Sig$ of a finite automaton that perform the same transformation $t$. In particular, he defined the length of the shortest such sequence to be the \emph{depth  with respect to $\Sig$} of the transformation $t$. The \emph{depth} of $t$ was then defined as the maximum over all $\Sig$ that produce $t$. 
Finally, he defined the \emph{complete depth} of a transformation to be its depth when $\Sig$ ranges over all alphabets  that generate all the transformations.
Many properties of the depth functions are established in this paper.

In 2003 and 2005, Krawetz, Lawrence and Shallit~\cite{KLS05} studied the state complexity of the operation ${\rm root}(L)=\{w\in\Sig^*\mid \exists n\ge 1 \text{ such that } w^n\in L\}$, which is bounded from above by $n^n$, where $n$ is the state complexity of $L$. In fact, they showed that a finite automaton with at most $n^n$ states can be constructed to accept ${\rm root}(L)$, and
 obtained a lower bound on the state complexity of ${\rm root}(L)$ for binary $L$.
 For alphabets of at least three letters, they showed that the bound on the state complexity of ${\rm root}(L)$
can be improved to $n^n-{n \choose 2}$.

\section{Ideal and Closed Languages}
If $w=uxv$ for some $u,v,x\in\Sigma^*$, then $u$ is a {\em prefix\/} of $w$, $v$ is a {\em suffix\/} of $w$, and  
 $x$ is a {\em factor\/} of $w$.
A prefix or suffix of $w$ is also a factor of $w$.

A~language $L$ is {\it prefix-convex\/}~\cite{AnBr09} if  $u, w \in L$ with
$u$ a prefix of $w$ implies that every word $v$ must also be in $L$ if
$u$ is a prefix of $v$ and $v$ is a prefix of $w$.
It is {\it prefix-closed\/} if $w\in L$ implies that
every prefix of $w$ is also in~$L$. 
In the same way, we define  \emph{suffix-convex} and \emph{factor-convex}, 
and the corresponding closed versions.

A language $L\subseteq\Sig^*$ is a \emph{right ideal} (respectively, \emph{left ideal}, \emph{two-sided ideal}) if it is non-empty and satisfies $L=L\Sig^*$ (respectively, $L=\Sig^*L$, $L=\Sig^*L\Sig^*$).
We refer to all three types  as \emph{ideal languages} or simply \emph{ideals}.

Suffix-closed languages were studied in 1974 by Gill and Kou~\cite{GiKo74}, in 1976 by Galil and Simon~\cite{GaSi76}, in 1979 by Veloso and Gill~\cite{VeGi79}, in 2001 by Holzer, K.~Salomaa, and Yu~\cite{HSY01}, in 2009 by Kao, Rampersad, and Shallit~\cite{KRS09} and by Ang and Brzozowski~\cite{AnBr09}, and in 2010 by Brzozowski, Jir\'askov\'a and Zou~\cite{BJZ10}.

Left and right ideals were studied  
by Paz and Peleg~\cite{PaPe65} in 1965 
under the names ``ultimate definite'' 
and ``reverse ultimate definite events". 
Complexity issues of conversion of nondeterministic finite automata to deterministic finite automata
in right, left, and two-sided ideals were studied 
in 2008 by Bordihn, Holzer, and Kutrib~\cite{BHK09}.
The closure properties of ideals were analyzed in~\cite{AnBr09}. 
Decision problems for various classes of convex languages, including ideals, were addressed by Brzozowski, Shallit and Xu in~\cite{BSX10}. 

\section{Transformations}
A \emph{transformation} of a set $Q$ is a mapping of $Q$ \emph{into} itself, whereas a \emph{permutation}
of $Q$ is a mapping of $Q$ \emph{onto} itself.
In this paper we consider only transformations of finite sets, and we assume
without loss of generality  that $Q=\{0,1,\ldots, n-1\}$.
An arbitrary transformation has the form
$$
t=\left( \begin{array}{ccccc}
0 & 1 &   \cdots &  n-2 & n-1 \\
i_0 & i_1 &   \cdots &  i_{n-2} & i_{n-1}
\end{array} \right ),
$$
where $i_k\in Q$ for $0\le k\le n-1$.
 To simplify the notation, such a transformation
will often be denoted  by $t:[i_0,  i_1,  \ldots,  i_{n-2}, i_{n-1}]$, or just $[i_0,  i_1,  \ldots,  i_{n-2}, i_{n-1}]$ if $t$ is understood.
The \emph{identity} transformation is the mapping
$$
t=\left( \begin{array}{ccccc}
0 & 1 &   \cdots &  n-2 & n-1 \\
0 & 1 &   \cdots &  {n-2} & {n-1}
\end{array} \right ).
$$
We will consider  \emph{cycles} of length $k$ of the following form:
$$
t=\left( \begin{array}{ccccccccccccc}
0 & 1 &   \cdots & i-1&\bf i &\bf  i+1 & \cdots &\bf  i+{k-2} &\bf  i+k-1 & i+k &\cdots & n-2 & n-1 \\
0 & 1 &   \cdots & i-1 &\bf i+1 &\bf i+2 &\cdots &\bf  i+k-1 &\bf i & i+k & \cdots& n-2 & n-1
\end{array} \right ),
$$
where we show in bold type the elements that are changed by $t$. To simplify the notation, such a cycle is represented by $(i,i+1,\ldots,i+k-1)$. 
A cycle of length~1 is the identity.
A \emph{singular} transformation is a transformation of the form
$$
t=\left( \begin{array}{cccccccc}
0 & 1 &   \cdots &i-1 &\bf i & i+1 &  n-2 & n-1 \\
0 & 1 &   \cdots & i-1 &\bf j & i+1 & n-2 & n-1
\end{array} \right ),
$$
which is denoted by $i\choose j$.
The singular transformation $i \choose i$ is the identity.
For $i<j$, a \emph{transposition} is a transformation of the form
$$
t=\left( \begin{array}{ccccccccccccc}
0 & 1 &   \cdots &i-1 &\bf i & i+1&\cdots & j-1 &\bf j &j+1 &\cdots &  n-2 & n-1 \\
0 & 1 &   \cdots & i-1 &\bf j & i+1 &\cdots & j-1 &\bf i &j+1 & \cdots &n-2 & n-1
\end{array} \right ),
$$
which is denoted by $(i,j)$, with $(i,i)$ being the identity.
A transposition is also a cycle of length 2.

A \emph{constant} transformation is a transformation of the form
$$
t=\left( \begin{array}{ccccc}
0 & 1 &   \cdots &  n-2 & n-1 \\
i & i &   \cdots &  i & i
\end{array} \right ),
$$
and it is denoted by $Q \choose i$.

The set of all $n^n$ transformations of a set $Q$ is a monoid under composition of transformations, with identity as the unit element. The set of all $n!$ permutations of $Q$ is a group, the \emph{symmetric} group of degree $n$. 
The following facts  about generators of particular semigroups are well-known:

\begin{theorem}[Permutations]
\label{thm:piccard}
The symmetric group $S_n$ of size $n!$ can be generated by any cyclic
 permutation of $n$ elements together with an arbitrary transposition. In particular, $S_n$ can be generated by
$c=(0,1,\ldots, n-1)$ and  $t=(0,1)$.
\end{theorem}

\begin{theorem}[Transformations]
\label{thm:salomaa}
The complete transformation monoid $T_n$ of size $n^n$ can be generated by any cyclic
permutation of $n$ elements together with a transposition and a ``returning'' transformation $r={n-1 \choose 0}$. In particular, $T_n$ can be generated by $c=(0,1,\ldots, n-1)$,  $t=(0,1)$ and $r={n-1 \choose 0}$.
\end{theorem}

\section{Quotient Complexity and Syntactic Complexity} 

If $\Sig$ is a non-empty finite alphabet, then $\Sig^*$ is the free monoid generated by $\Sig$, and $\Sig^+$ is the free semigroup generated by $\Sig$.  A \emph{word} is any element of $\Sig^*$, and the empty word is $\eps$. The length of a word $w\in \Sig^*$ is $|w|$.
A \emph{language} over $\Sig$ is any subset of $\Sig^*$. 
The \emph{left quotient}, or simply \emph{quotient,} of a language $L$ by a word $w$ is  the language $L_w=\{x\in \Sig^*\mid wx\in L \}$.

An equivalence relation $\sim$ on $\Sig^*$ is a \emph{left congruence}  if, for all $x,y\in \Sig^*$, 
\begin{equation}
x\sim y {\Lra} ux \sim uy, \mbox{ for all } u\in\Sig^*.
\end{equation}
It   is a \emph{right congruence} if, for all $x,y\in \Sig^*$,
\begin{equation}
x\sim y {\Lra} xv \sim yv, \mbox{ for all } v\in\Sig^*.
\end{equation}
It is a \emph{congruence} if it is both a left and a right congruence. 
Equivalently, $\sim$ is a congruence if 
\begin{equation}
x\sim y {\Lra} uxv \sim uyv, \mbox{ for all } u,v\in\Sig^*.
\end{equation}

For any language $L\subseteq \Sig^*$, define the \emph{Nerode congruence}~\cite{Ner58} $\raL$ of $L$ by
\begin{equation}
x \raL y \mbox{ if and only if } xv\in L  \Leftrightarrow yv\in L, \mbox { for all } u,v\in\Sig^*.
\end{equation}
Evidently, $L_x=L_y$ if and only if $x\raL y$.
Thus, each equivalence class of this congruence corresponds to a distinct quotient of $L$.

For any language $L\subseteq \Sig^*$, define the \emph{Myhill congruence}~\cite{Myh57} $\lraL$ of $L$ by
\begin{equation}
x \lraL y \mbox{ if and only if } uxv\in L  \Leftrightarrow uyv\in L\mbox { for all } u,v\in\Sig^*.
\end{equation}
This congruence is also known as the \emph{syntactic congruence} of $L$.
The  semigroup  $\Sig^+/ \lraL$ of equivalence classes of the relation $\lraL$,  is the \emph{syntactic semigroup} of $L$, and 
$\Sig^*/ \lraL$ is the \emph{syntactic monoid} of~$L$. 
The \emph{syntactic complexity} $\sig(L)$ of $L$ is the cardinality of its syntactic semigroup.
The \emph{monoid complexity} $\mu(L)$ of $L$ is the cardinality of its syntactic monoid.
If the  equivalence containing $\eps$ is a singleton in the syntactic monoid, then $\sig(L)=\mu(L)-1$; otherwise, $\sig(L)=\mu(L)$.

A \emph{(deterministic) semiautomaton} is a triple, $\cS=(Q, \Sig, \delta)$, where
$Q$ is a finite, non-empty set of \emph{states}, $\Sig$ is a finite non-empty \emph{alphabet}, and $\delta:Q\times \Sig\to Q$ is the \emph{transition function}. 
A~\emph{deterministic finite automaton} or simply \emph{automaton} is a quintuple $\cA=(Q, \Sig, \delta, q_0,F)$, where $Q$, $\Sig$, and $\delta$ are as defined in the semiautomaton $\cS=(Q, \Sig, \delta)$,  $q_0\in Q$ is the \emph{initial state}, and $F\subseteq Q$ is the set of \emph{final states}.
A~\emph{nondeterministic finite automaton}  or simply \emph{nondeterministic automaton} is a quintuple $\cN=(Q, \Sig, \eta, S,F)$, where $Q$, $\Sig$, and $F$ are as defined in  a deterministic automaton,  $S\subseteq  Q$ is the \emph{set of initial states}, and $\eta:Q\times \Sig\to 2^Q$ is the  \emph{transition function}.

The \emph{$\eps$-function}  $L^\eps$ of a regular language $L$ is
$L^\eps=\emp$ if $\eps\not\in L$; $L^\eps=\eps$ if $\eps\in L$.
The \emph{quotient automaton} of a regular language $L$ is 
$\cA=(Q, \Sig, \delta, q_0,F)$, where $Q=\{L_w\mid w\in\Sig^*\}$, $\delta(L_w,a)=L_{wa}$, 
$q_0=L_\eps=L$,  $F=\{L_w \mid L_w^\eps=\eps\}$, and $L_w^\eps=(L_w)^\eps$.
The number of states in the quotient automaton of $L$ is the \emph{quotient complexity} of $L$. 
The quotient complexity is the same as the state complexity, but there are advantages to using quotients~\cite{Brz10}.
A~quotient automaton can be conveniently represented by \emph{quotient equations}~\cite{Brz64}:
\begin{equation}
L_w=\bigcup_{a\in \Sig}aL_{wa}\cup L_w^\eps,
\end{equation}
where there is one equation for each distinct quotient $L_w$.

In terms of automata, each equivalence class $[w]_{\raL}$ of $\raL$ is the set of all words $w$ that take the automaton to the same state from the initial state. In terms of quotients, it is the set of words $w$ that can all be followed by the same quotient $L_w$.

In terms of automata, each equivalence class $[w]_\lraL$ of the syntactic congruence is the set of all words that perform the same transformation on the set of states.

The \emph{transformation semigroup} (respectively, \emph{transformation monoid}) of an automaton is the set of transformations performed by words of $\Sig^+$ (respectively, $\Sig^*$) on the set of states. The transformation semigroup (monoid) of the quotient automaton of $L$ is isomorphic to the syntactic semigroup (monoid) of~$L$.

\begin{proposition}
\label{prop:basicbounds}
For any language $L$ with $\kappa(L)=n$, we have $n-1\le \sig(L)\le n^n$. 
\end{proposition}

\begin{proof}
Since every state other than the initial state has to be reachable from the initial state by a non-empty word, there must be at least $n-1$ transformations.
If $\Sig=\{a\}$ and $L=a^{n-1}a^*$, then  $\kappa(L)=n$, and $\sigma(L)=n-1$.
Thus the lower bound $n-1$ is achievable. It is evident that $n^n$ is an upper bound, and by Theorem~\ref{thm:salomaa} this upper bound is achievable if $|\Sig|\ge 3$.
\qed
\end{proof}

If one of the quotients of $L$ is $\emp$ (respectively, $\eps$, $\Sig^*$, $\Sig^+$), then we say that $L$ \emph{has $\emp$} (respectively, $\eps$, $\Sig^*$, $\Sig^+$).
A quotient $L_w$ of a language $L$ is \emph{\ur{} (ur)}~\cite{Brz10} if $L_x=L_w$ implies that $x=w$.
If $L_{wa}$ is \ur{} for $a\in \Sig$,  then so is $L_w$. 
Thus, if $L$ has a \ur{} quotient, then $L$ itself is \ur{} by the empty word, \ie, the minimal automaton of $L$ is \emph{non-returning}.

\begin{theorem}[Special Quotients]
\label{thm:specialquotients}
Let $L$ be any language with $\kappa(L)=n$. \\
1. If $L$ has $\emp$ or $\Sig^*$, then $\sig(L)\le n^{n-1}$.\\
2. If $L$ has $\eps$ or $\Sig^+$, then $\sig(L)\le n^{n-2}$.\\
3. If $L$ is uniquely reachable, then $\sig(L)\le (n-1)^n$.\\
4. If $L_a$ is uniquely reachable for some $a\in\Sig$, then $\sig(L)\le 1+ (n-2)^n$.\\
Moreover, these effects are cumulative as shown in Table~\ref{tab:special}.
\begin{table}[ht]
\caption{Upper bounds on syntactic complexity for languages with special quotients.}
\begin{center}
$
\begin{array}{| c | c | c | c | c | c | c |}   
\hline
\hspace{.25cm}\emp \hspace{.25cm} &\hspace{.1cm}\Sig^*\hspace{.1cm} & 
\hspace{.25cm}\eps \hspace{.25cm}&  \hspace{.1cm}\Sig^+\hspace{.1cm}  &  &  \mbox{$L$  is ur} &  \mbox{ $L_a$ is ur}
\\
\hline
\surd & &    &    & n^{n-1}  & (n-1)^{n-1} & 1+ (n-3)^{n-2}
\\
\hline
 & \surd &    &    & n^{n-1} & (n-1)^{n-1} & 1+ (n-3)^{n-2}
\\
\hline
\surd  &  & \surd  &    & n^{n-2} & (n-1)^{n-2} & 1+ (n-4)^{n-2}
\\
\hline
 & \surd  &    & \surd   &  n^{n-2} & (n-1)^{n-2} & 1+ (n-4)^{n-2}
\\
\hline
 \surd  & \surd   &    &  &  n^{n-2} & (n-1)^{n-2} & 1+ (n-4)^{n-2}
\\
\hline
 \surd  & \surd &      &\surd  &  n^{n-3} & (n-1)^{n-3} & 1+ (n-5)^{n-2}
\\
\hline
 \surd  &\surd   & \surd   &  & n^{n-3} & (n-1)^{n-3} & 1+ (n-5)^{n-2}
\\
\hline
 \surd  &\surd   & \surd   &\surd  & n^{n-4} & (n-1)^{n-4} & 1+ (n-6)^{n-2}
 \\
\hline
\end{array}
$
\end{center}
\label{tab:special}
\end{table}
\end{theorem}

\begin{proof}
1. Since $\emp_a=\emp$ for all $a\in\Sig$, there are only $n-1$ states in the quotient automaton with which one can distinguish two transformations. Hence there are at most $n^{n-1}$ such transformations. If $L$ has $\Sig^*$, then $\Sig^*_a=\Sig^*$, for all $a\in\Sig$, and the same argument applies.

2. Since $\eps_a=\emp$ for all $a\in\Sig$, $L$ has $\emp$ if $L$ has $\eps$. Now there are two states that do not contribute to distinguishing among different transformations. 
Dually, $\Sig^+_a=\Sig^*$ for all $a\in\Sig$, and the same argument applies.

3.  If $L$ is uniquely reachable then $L_w= L$ implies $w=\eps$. Thus $L$ does not appear as a result of any transformation by a word in $\Sig^+$, and there remain only $n-1$ choices for each of the $n$ states. 

4. If $L_a$ is \ur{}, then so is $L$. Hence $L$ never appears as a result of a transformation by a word in $\Sig^+$, and $L_a$ appears only in one transformation. Therefore there can be at most $(n-2)^n$ other transformations.
\qed
\end{proof}

\section{Right Ideals and Prefix-Closed Languages}
In this section we characterize the syntactic complexity of right ideals. The automaton defined below plays an important role in this theory.
\begin{definition}
\label{def:RightIdeal}
For $n\ge 4$, define the automaton 
$$\cA_n=(\{0,1,\ldots,n-1\}, \{a,b,c,d\},\delta,0, \{n-1\}),$$ 
where $a=(0,1,\ldots,n-2)$, 
$b=(0,1)$, 
$c={n-2 \choose 0}$,
and $d={n-2 \choose n-1}$.
The transition function
$\delta$ is then defined using these transformations.
The automaton so defined accepts a right ideal and is minimal; it is depicted in Fig.~\ref{fig:RightIdeal}.
\end{definition}
\begin{figure}[hbt]
\begin{center}
\setlength{\unitlength}{0.00052493in}
\begingroup\makeatletter\ifx\SetFigFont\undefined%
\gdef\SetFigFont#1#2#3#4#5{%
  \reset@font\fontsize{#1}{#2pt}%
  \fontfamily{#3}\fontseries{#4}\fontshape{#5}%
  \selectfont}%
\fi\endgroup%
{\renewcommand{\dashlinestretch}{30}
\begin{picture}(7725,2080)(0,-10)
\put(3334,147){\makebox(0,0)[lb]{\smash{{\SetFigFont{8}{9.6}{\familydefault}{\mddefault}{\updefault}$a,c$}}}}
\put(6257.000,1577.333){\arc{333.333}{2.2143}{7.2105}}
\blacken\path(6394.638,1540.417)(6357.000,1444.000)(6435.107,1511.913)(6394.638,1540.417)
\put(2747.000,1577.333){\arc{333.333}{2.2143}{7.2105}}
\blacken\path(2884.638,1540.417)(2847.000,1444.000)(2925.107,1511.913)(2884.638,1540.417)
\put(1667.000,1577.333){\arc{333.333}{2.2143}{7.2105}}
\blacken\path(1804.638,1540.417)(1767.000,1444.000)(1845.107,1511.913)(1804.638,1540.417)
\put(5187.000,1547.333){\arc{333.333}{2.2143}{7.2105}}
\blacken\path(5324.638,1510.417)(5287.000,1414.000)(5365.107,1481.913)(5324.638,1510.417)
\put(7362.000,1637.333){\arc{333.333}{2.2143}{7.2105}}
\blacken\path(7499.638,1600.417)(7462.000,1504.000)(7540.107,1571.913)(7499.638,1600.417)
\put(591,1129){\ellipse{630}{630}}
\put(2755,1129){\ellipse{630}{630}}
\put(1685,1127){\ellipse{630}{630}}
\put(6275,1129){\ellipse{630}{630}}
\put(7354,1137){\ellipse{630}{630}}
\put(5187,1129){\ellipse{630}{630}}
\put(7357,1138){\ellipse{720}{720}}
\path(12,1129)(282,1129)
\blacken\path(162.000,1099.000)(282.000,1129.000)(162.000,1159.000)(162.000,1099.000)
\path(1992,1129)(2442,1129)
\blacken\path(2322.000,1099.000)(2442.000,1129.000)(2322.000,1159.000)(2322.000,1099.000)
\path(3072,1129)(3522,1129)
\blacken\path(3402.000,1099.000)(3522.000,1129.000)(3402.000,1159.000)(3402.000,1099.000)
\path(912,1129)(1362,1129)
\blacken\path(1242.000,1099.000)(1362.000,1129.000)(1242.000,1159.000)(1242.000,1099.000)
\path(6589,1129)(6987,1129)
\blacken\path(6867.000,1099.000)(6987.000,1129.000)(6867.000,1159.000)(6867.000,1099.000)
\path(1407,971)(852,971)
\blacken\path(972.000,1001.000)(852.000,971.000)(972.000,941.000)(972.000,1001.000)
\path(4422,1137)(4872,1137)
\blacken\path(4752.000,1107.000)(4872.000,1137.000)(4752.000,1167.000)(4752.000,1107.000)
\path(5495,1129)(5945,1129)
\blacken\path(5825.000,1099.000)(5945.000,1129.000)(5825.000,1159.000)(5825.000,1099.000)
\path(6124,844)(6123,844)(6122,842)
	(6119,840)(6114,837)(6108,833)
	(6099,827)(6089,819)(6075,810)
	(6060,800)(6042,787)(6021,774)
	(5999,759)(5974,743)(5947,726)
	(5919,708)(5889,689)(5857,670)
	(5825,650)(5790,629)(5754,609)
	(5717,588)(5679,567)(5638,545)
	(5596,524)(5552,502)(5505,480)
	(5457,457)(5405,434)(5351,411)
	(5293,387)(5233,363)(5169,339)
	(5102,315)(5033,291)(4962,267)
	(4890,244)(4820,223)(4752,203)
	(4688,185)(4628,169)(4574,154)
	(4525,141)(4482,129)(4444,119)
	(4411,110)(4383,103)(4359,96)
	(4339,91)(4322,86)(4307,81)
	(4294,78)(4282,74)(4271,71)
	(4260,69)(4248,66)(4235,63)
	(4220,60)(4202,57)(4181,54)
	(4156,51)(4127,47)(4093,44)
	(4053,39)(4008,35)(3957,31)
	(3900,27)(3837,22)(3769,19)
	(3697,15)(3621,13)(3544,12)
	(3467,12)(3391,14)(3320,16)
	(3253,19)(3191,23)(3136,26)
	(3087,30)(3044,33)(3006,37)
	(2975,40)(2948,43)(2926,46)
	(2908,48)(2893,51)(2881,53)
	(2871,55)(2862,57)(2854,60)
	(2846,62)(2837,65)(2827,68)
	(2815,72)(2800,76)(2782,82)
	(2760,88)(2733,95)(2702,103)
	(2664,113)(2621,125)(2572,138)
	(2517,152)(2456,169)(2389,187)
	(2317,207)(2242,229)(2165,252)
	(2092,275)(2019,299)(1949,322)
	(1880,346)(1814,370)(1751,394)
	(1690,417)(1632,440)(1576,463)
	(1522,486)(1470,508)(1420,530)
	(1372,552)(1325,574)(1279,595)
	(1235,616)(1192,637)(1150,658)
	(1110,678)(1071,698)(1033,717)
	(998,736)(964,753)(932,770)
	(903,786)(876,800)(852,813)
	(830,825)(811,835)(796,844)
	(783,851)(773,857)(765,861)(754,867)
\blacken\path(873.713,835.875)(754.000,867.000)(844.982,783.201)(873.713,835.875)
\put(536,1062){\makebox(0,0)[lb]{\smash{{\SetFigFont{8}{9.6}{\rmdefault}{\mddefault}{\updefault}$0$}}}}
\put(1624,1062){\makebox(0,0)[lb]{\smash{{\SetFigFont{8}{9.6}{\rmdefault}{\mddefault}{\updefault}$1$}}}}
\put(2711,1062){\makebox(0,0)[lb]{\smash{{\SetFigFont{8}{9.6}{\rmdefault}{\mddefault}{\updefault}$2$}}}}
\put(7104,1074){\makebox(0,0)[lb]{\smash{{\SetFigFont{7}{8.4}{\familydefault}{\mddefault}{\updefault}$n-1$}}}}
\put(6010,1066){\makebox(0,0)[lb]{\smash{{\SetFigFont{7}{8.4}{\familydefault}{\mddefault}{\updefault}$n-2$}}}}
\put(2120,1227){\makebox(0,0)[lb]{\smash{{\SetFigFont{8}{9.6}{\familydefault}{\mddefault}{\updefault}$a$}}}}
\put(3192,1249){\makebox(0,0)[lb]{\smash{{\SetFigFont{8}{9.6}{\familydefault}{\mddefault}{\updefault}$a$}}}}
\put(5622,1249){\makebox(0,0)[lb]{\smash{{\SetFigFont{8}{9.6}{\familydefault}{\mddefault}{\updefault}$a$}}}}
\put(957,1264){\makebox(0,0)[lb]{\smash{{\SetFigFont{8}{9.6}{\familydefault}{\mddefault}{\updefault}$a,b$}}}}
\put(4917,1084){\makebox(0,0)[lb]{\smash{{\SetFigFont{7}{8.4}{\familydefault}{\mddefault}{\updefault}$n-3$}}}}
\put(409,1834){\makebox(0,0)[lb]{\smash{{\SetFigFont{8}{9.6}{\familydefault}{\mddefault}{\updefault}$c,d$}}}}
\put(1504,1842){\makebox(0,0)[lb]{\smash{{\SetFigFont{8}{9.6}{\familydefault}{\mddefault}{\updefault}$c,d$}}}}
\put(2517,1842){\makebox(0,0)[lb]{\smash{{\SetFigFont{8}{9.6}{\familydefault}{\mddefault}{\updefault}$b,c,d$}}}}
\put(4909,1834){\makebox(0,0)[lb]{\smash{{\SetFigFont{8}{9.6}{\familydefault}{\mddefault}{\updefault}$b,c,d$}}}}
\put(6949,1894){\makebox(0,0)[lb]{\smash{{\SetFigFont{8}{9.6}{\familydefault}{\mddefault}{\updefault}$a,b,c,d$}}}}
\put(3875,1070){\makebox(0,0)[lb]{\smash{{\SetFigFont{8}{9.6}{\familydefault}{\mddefault}{\updefault}$\cdots$}}}}
\put(1136,754){\makebox(0,0)[lb]{\smash{{\SetFigFont{8}{9.6}{\familydefault}{\mddefault}{\updefault}$b$}}}}
\put(4550,1242){\makebox(0,0)[lb]{\smash{{\SetFigFont{8}{9.6}{\familydefault}{\mddefault}{\updefault}$a$}}}}
\put(6184,1857){\makebox(0,0)[lb]{\smash{{\SetFigFont{8}{9.6}{\familydefault}{\mddefault}{\updefault}$b$}}}}
\put(6732,1249){\makebox(0,0)[lb]{\smash{{\SetFigFont{8}{9.6}{\familydefault}{\mddefault}{\updefault}$d$}}}}
\put(592.000,1562.333){\arc{333.333}{2.2143}{7.2105}}
\blacken\path(729.638,1525.417)(692.000,1429.000)(770.107,1496.913)(729.638,1525.417)
\end{picture}
}
\end{center}
\caption{Automaton $\cA_n$ of a right ideal with $n^{n-1}$ transformations.} 
\label{fig:RightIdeal}
\end{figure}
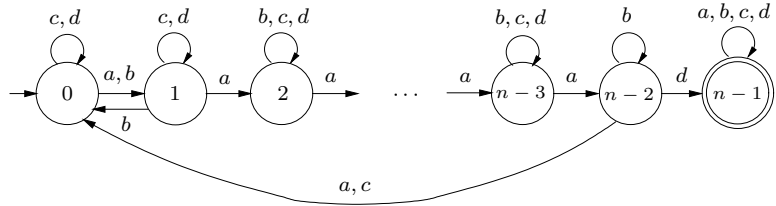

\begin{theorem}[Right Ideals and Prefix-Closed Languages)]
\label{thm:rightideals}
Let $L\subseteq\Sig^*$ have quotient complexity $n$.
If $L$ is a right ideal or a prefix-closed language, then the syntactic complexity of $L$ is less than or equal to $ n^{n-1}$. Moreover, the bound is tight for $n=1$  if $|\Sig|\ge 1$, for $n=2$ if $|\Sig|\ge 2$, for $n=3$ if $|\Sig|\ge 3$, and for $n\ge 4$  if $|\Sig|\ge 4$.
\end{theorem}
\begin{proof}
Since every prefix-closed language other than $\Sigma^*$ is the complement of a right ideal, and complementation preserves syntactic complexity, it suffices to consider only right ideals.

If $L$ is a right ideal, then it has $\Sig^*$ as a quotient. By Theorem~\ref{thm:specialquotients}, we have
$\sig(L)\le n^{n-1}$. 

Next we prove that the language $L=L(\cA_n)$ accepted by the automaton of Fig.~\ref{fig:RightIdeal}
meets this bound.
Consider any transformation $t$ of the form
$$
t=\left( \begin{array}{ccccccc}
0 & 1 & 2 &  \cdots & n-3 & n-2 & n-1 \\
i_0 & i_1 & i_2 &  \cdots & i_{n-3} & i_{n-2} & n-1
\end{array} \right ),
$$
where $i_k\in\{0,1,\ldots, n-1\}$ for $0\le k\le n-2$.
There are two cases:
\be
\item Suppose $i_k\ne n-1$ for all $k$, $0\le k\le n-2$. By Theorem~\ref{thm:salomaa}, 
since all the images of the first $n-1$ states are in the set $\{0,1,\ldots, n-2\}$, transformation $t$ can be performed by $\cA_n$.
\item If $i_h = n-1$ for some $h$, $0\le h\le n-2$, then by the pigeon-hole principle, there exists some $j$,
$0\le j\le n-2$ such that $i_k\ne j$ for all $k$, $0\le k\le n-2$.
\ee

Define $i'_k$ for all $0\le k\le n-2$ as follows:
\begin{equation*} 
i'_k = \left\{ 
\begin{array}{cc} 
j,	&  \quad \text{if } i_k=n-1;\\ 
i_k,   &  \quad  \text{if } i_k\ne n-1.\\ 
\end{array} \right. 
\end{equation*}
Then let 
$$
s=\left( \begin{array}{cccccccc}
0 & 1 & 2 & 3 & \cdots & n-3 & n-2 & n-1 \\
i'_0 & i'_1 & i'_2 & i'_3 & \cdots & i'_{n-3} & i'_{n-2} & n-1
\end{array} \right ),$$
Also, let $r=(j,n-2)$.
Since all the images of the first $n-1$ states in $s$ and $r$ are in the set $\{0,1,\ldots, n-2\}$,
by Theorem~\ref{thm:salomaa}, $s$ and $r$ can be performed by $\cA_n$.

We show now that $t=srdr$, which implies that $t$ can also be performed by~$\cA_n$.
If $t$ maps $k$ to $n-1$, then $s$ maps $k$ to $j$, $r$ maps $j$ to $n-2$, $d$ maps $n-2$ to $n-1$, and $r$ maps $n-1$ to $n-1$.
If $t$ maps $k$ to $n-2$, then $s$ maps $k$ to $n-2$, $r$ maps $n-2$ to $j$, $d$ maps $j$ to $j$, and $r$ maps $j$ to $n-2$.
If $t$ maps $k$ to $i_k < n-2$, then so does $srdr$.
Hence in all cases the mapping performed by $t$ is the same as that of $srdr$.

Since there are $n^{n-1}$ transformations like $t$,  $L(\cA_n)$ meets the bound.

Now we consider the values $n\le 5$. The bounds claimed below have all been verified by a computer program.

\bd
\item[n=1:]
There is only one type of right ideal with $n=1$, namely  $L=\Sig^*$, and its syntactic complexity is $\sig(L)=1$. Thus the bound $1^0=1$ is tight  for $|\Sig|\ge 1$.

\item[n=2:]
If $|\Sig|=1$,  there is only one  right ideal, $L=aa^*$, and $\sig(L)=1$.\\
If $|\Sig|=2$, then  $b^*a(a+b)^*$ meets the bound $2^1=2$ of the theorem.

\item[n=3:]
If $|\Sig|=1$, there is only one right ideal, $L=aaa^*$, and $\sig(L)=2$. 

For $n=3$,  inputs $a$ and $b$ of the automaton of Fig.~\ref{fig:RightIdeal} coincide.

If $|\Sig|=2$, we have verified that $\sig(L)\le 7$ for all right ideals, and the language of $\cA_3$  restricted to input alphabet $\{a,d\}$ meets the bound 7.

If $|\Sig|=3$, then the language of $\cA_3$ restricted to input alphabet $\{a,c,d\}$ meets the bound $3^2=9$ of the theorem.

\item[n= 4:]
If $|\Sig|=1$, there is only one right ideal, $L=aaaa^*$, and $\sig(L)=3$.\\
For $|\Sig|=2$, we have verified that $\sig(L)\le 31$  for all right ideals $L$. 
The bound is reached with the inputs $a:[1,2,0,3]$ and $b:[1,0,3,3]$.
\\
For $|\Sig|=3$, we have verified that $\sig(L)\le 61$ for all right ideals $L$, and
$\cA_4$ restricted to input alphabet $\{a,c,d\}$ meets this bound.
\item[n=5:]
For $|\Sig|=2$, we have verified that $\sig(L)\le 167$  for all right ideals $L$. 
The bound is reached with the inputs $a:[0,1,0,2,4]$ and $b:[1,3,2,4,4]$, or with
$a:[0,0,1,2,4]$ and $b:[2,3,0,4.4]$.
For $|\Sig|=3$, we have verified that $\sig(L)\le 545$  for all right ideals. 
The bound is reached with the inputs $a:[0,0,1,3,4]$, $b:[2,0,3,1,4]$, and $c:[3,1,2,4,4]$.
\qed
\ed
\end{proof}

Table~\ref{tab:RISummary} summarizes our result for right ideals.
All the numbers shown are tight upper bounds.  In general, there are many solutions with the same complexity.
\begin{table}[ht]
\caption{Syntactic complexity bounds for right ideals.}
\label{tab:RISummary}
\begin{center}
$
\begin{array}{| c ||c|c| c| c|c|c|c|}    
\hline
\ \  \ \ &\ \ n=1 \ \ &\ \ n=2 \ \ &\ \ n=3 \  \ & \   n=4 \ \ 
&  \ n=5 \ &\ \ \ldots \ \ 
&\  n=n \
  \\
\hline  \hline
  |\Sig|=1
& \bf 1&\bf 1  &\bf 2&\bf  3  	&\bf 4 &   \ldots	& \bf n-1 \\
\hline
|\Sig|=2 
& - & \bf 2 &\bf 7 &\bf 31  		&\bf 167 & \ldots  	&  \\
\hline
|\Sig|=3
& - &-  &\bf 9  &\bf 61 	&\bf 545 &   \ldots   	& \\
\hline
|\Sig|=4
& - &-  & -  &\bf 64 	&\bf 625 & \ldots  	&\bf n^{n-1}  \\
\hline
\hline
\end{array}
$
\end{center}
\end{table}

It is interesting to note that for our right ideal $L$ with maximal syntactic complexity, the reverse language has maximal state complexity. Recall that the reverse $w^R$ of a word $w$ is defined inductively by
$\eps^R=\eps$, $(au)^R=u^Ra$. The reverse of a language $L$ is $L^R=\{w^R\mid w\in R\}$.
It was shown in~\cite{BJL10} that the reverse of a right ideal with $n$ quotients has at most $2^{n-1}$ quotients, and that this bound can be met by a binary automaton. 
We now prove that automaton $\cA_n'$, which is $\cA_n$ restricted to inputs $a$ and $d$,  is another example of a binary automaton that meets the $2^{n-1}$ bound for reversal of right ideals.
 The nondeterministic automaton  $\cN_n$ obtained by reversing  $\cA_n'$ is shown in Fig.~\ref{fig:RightIdealB}. 
 
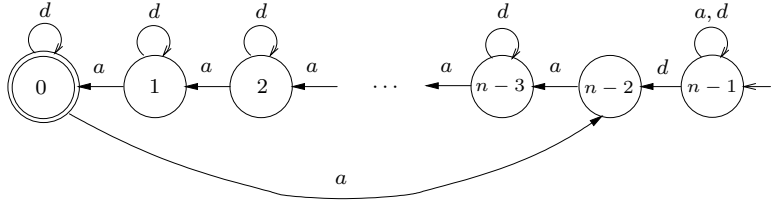
\begin{figure}[h]
\begin{center}
\setlength{\unitlength}{0.00052493in}
\begingroup\makeatletter\ifx\SetFigFont\undefined%
\gdef\SetFigFont#1#2#3#4#5{%
  \reset@font\fontsize{#1}{#2pt}%
  \fontfamily{#3}\fontseries{#4}\fontshape{#5}%
  \selectfont}%
\fi\endgroup%
{\renewcommand{\dashlinestretch}{30}
\begin{picture}(7715,2035)(0,-10)
\put(5475,1286){\makebox(0,0)[lb]{\smash{{\SetFigFont{8}{9.6}{\familydefault}{\mddefault}{\updefault}$a$}}}}
\put(1476.000,1592.000){\arc{332.800}{2.2155}{7.2093}}
\path(1613.379,1555.062)(1576.000,1459.000)(1654.187,1526.170)
\put(4996.000,1562.000){\arc{332.800}{2.2155}{7.2093}}
\path(5133.379,1525.062)(5096.000,1429.000)(5174.187,1496.170)
\put(7103.000,1577.000){\arc{332.800}{2.2155}{7.2093}}
\path(7240.379,1540.062)(7203.000,1444.000)(7281.187,1511.170)
\put(381.000,1614.000){\arc{332.800}{2.2155}{7.2093}}
\path(518.379,1577.062)(481.000,1481.000)(559.187,1548.170)
\put(2564,1144){\ellipse{630}{630}}
\put(1494,1142){\ellipse{630}{630}}
\put(6084,1144){\ellipse{630}{630}}
\put(4996,1144){\ellipse{630}{630}}
\put(7116,1147){\ellipse{630}{630}}
\put(368,1139){\ellipse{720}{720}}
\put(365,1134){\ellipse{630}{630}}
\blacken\path(1921.000,1174.000)(1801.000,1144.000)(1921.000,1114.000)(1921.000,1174.000)
\path(1801,1144)(2251,1144)
\blacken\path(3001.000,1174.000)(2881.000,1144.000)(3001.000,1114.000)(3001.000,1174.000)
\path(2881,1144)(3331,1144)
\blacken\path(841.000,1174.000)(721.000,1144.000)(841.000,1114.000)(841.000,1174.000)
\path(721,1144)(1171,1144)
\blacken\path(6518.000,1174.000)(6398.000,1144.000)(6518.000,1114.000)(6518.000,1174.000)
\path(6398,1144)(6796,1144)
\blacken\path(4351.000,1182.000)(4231.000,1152.000)(4351.000,1122.000)(4351.000,1182.000)
\path(4231,1152)(4681,1152)
\blacken\path(5424.000,1174.000)(5304.000,1144.000)(5424.000,1114.000)(5424.000,1174.000)
\path(5304,1144)(5754,1144)
\path(7703,1144)(7433,1144)
\path(7553.000,1174.000)(7433.000,1144.000)(7553.000,1114.000)
\blacken\path(5919.896,750.608)(6001.000,844.000)(5885.488,799.761)(5919.896,750.608)
\path(6001,844)(5991,837)(5985,833)
	(5976,827)(5966,819)(5952,810)
	(5937,800)(5919,787)(5898,774)
	(5876,759)(5851,743)(5824,726)
	(5796,708)(5766,689)(5734,670)
	(5702,650)(5667,629)(5631,609)
	(5594,588)(5556,567)(5515,545)
	(5473,524)(5429,502)(5382,480)
	(5334,457)(5282,434)(5228,411)
	(5170,387)(5110,363)(5046,339)
	(4979,315)(4910,291)(4839,267)
	(4767,244)(4697,223)(4629,203)
	(4565,185)(4505,169)(4451,154)
	(4402,141)(4359,129)(4321,119)
	(4288,110)(4260,103)(4236,96)
	(4216,91)(4199,86)(4184,81)
	(4171,78)(4159,74)(4148,71)
	(4137,69)(4125,66)(4112,63)
	(4097,60)(4079,57)(4058,54)
	(4033,51)(4004,47)(3970,44)
	(3930,39)(3885,35)(3834,31)
	(3777,27)(3714,22)(3646,19)
	(3574,15)(3498,13)(3421,12)
	(3344,12)(3268,14)(3197,16)
	(3130,19)(3068,23)(3013,26)
	(2964,30)(2921,33)(2883,37)
	(2852,40)(2825,43)(2803,46)
	(2785,48)(2770,51)(2758,53)
	(2748,55)(2739,57)(2731,60)
	(2723,62)(2714,65)(2704,68)
	(2692,72)(2677,76)(2659,82)
	(2637,88)(2610,95)(2579,103)
	(2541,113)(2498,125)(2449,138)
	(2394,152)(2333,169)(2266,187)
	(2194,207)(2119,229)(2042,252)
	(1969,275)(1896,299)(1826,322)
	(1757,346)(1691,370)(1628,394)
	(1567,417)(1509,440)(1453,463)
	(1399,486)(1347,508)(1297,530)
	(1249,552)(1202,574)(1156,595)
	(1112,616)(1069,637)(1027,658)
	(987,678)(948,698)(910,717)
	(875,736)(841,753)(809,770)
	(780,786)(753,800)(729,813)
	(707,825)(688,835)(673,844)
	(660,851)(650,857)(642,861)
	(637,864)(634,866)(632,867)(631,867)
\put(1433,1077){\makebox(0,0)[lb]{\smash{{\SetFigFont{8}{9.6}{\rmdefault}{\mddefault}{\updefault}$1$}}}}
\put(2520,1077){\makebox(0,0)[lb]{\smash{{\SetFigFont{8}{9.6}{\rmdefault}{\mddefault}{\updefault}$2$}}}}
\put(5819,1081){\makebox(0,0)[lb]{\smash{{\SetFigFont{7}{8.4}{\familydefault}{\mddefault}{\updefault}$n-2$}}}}
\put(4726,1099){\makebox(0,0)[lb]{\smash{{\SetFigFont{7}{8.4}{\familydefault}{\mddefault}{\updefault}$n-3$}}}}
\put(3684,1085){\makebox(0,0)[lb]{\smash{{\SetFigFont{8}{9.6}{\familydefault}{\mddefault}{\updefault}$\cdots$}}}}
\put(1425,1841){\makebox(0,0)[lb]{\smash{{\SetFigFont{8}{9.6}{\familydefault}{\mddefault}{\updefault}$d$}}}}
\put(871,1279){\makebox(0,0)[lb]{\smash{{\SetFigFont{8}{9.6}{\familydefault}{\mddefault}{\updefault}$a$}}}}
\put(1951,1286){\makebox(0,0)[lb]{\smash{{\SetFigFont{8}{9.6}{\familydefault}{\mddefault}{\updefault}$a$}}}}
\put(3016,1286){\makebox(0,0)[lb]{\smash{{\SetFigFont{8}{9.6}{\familydefault}{\mddefault}{\updefault}$a$}}}}
\put(4373,1286){\makebox(0,0)[lb]{\smash{{\SetFigFont{8}{9.6}{\familydefault}{\mddefault}{\updefault}$a$}}}}
\put(3308,162){\makebox(0,0)[lb]{\smash{{\SetFigFont{8}{9.6}{\familydefault}{\mddefault}{\updefault}$a$}}}}
\put(2505,1841){\makebox(0,0)[lb]{\smash{{\SetFigFont{8}{9.6}{\familydefault}{\mddefault}{\updefault}$d$}}}}
\put(4942,1841){\makebox(0,0)[lb]{\smash{{\SetFigFont{8}{9.6}{\familydefault}{\mddefault}{\updefault}$d$}}}}
\put(6861,1089){\makebox(0,0)[lb]{\smash{{\SetFigFont{7}{8.4}{\familydefault}{\mddefault}{\updefault}$n-1$}}}}
\put(322,1849){\makebox(0,0)[lb]{\smash{{\SetFigFont{8}{9.6}{\familydefault}{\mddefault}{\updefault}$d$}}}}
\put(286,1070){\makebox(0,0)[lb]{\smash{{\SetFigFont{8}{9.6}{\rmdefault}{\mddefault}{\updefault}$0$}}}}
\put(6937,1841){\makebox(0,0)[lb]{\smash{{\SetFigFont{8}{9.6}{\familydefault}{\mddefault}{\updefault}$a,d$}}}}
\put(6556,1264){\makebox(0,0)[lb]{\smash{{\SetFigFont{8}{9.6}{\familydefault}{\mddefault}{\updefault}$d$}}}}
\put(2556.000,1592.000){\arc{332.800}{2.2155}{7.2093}}
\path(2693.379,1555.062)(2656.000,1459.000)(2734.187,1526.170)
\end{picture}
}
\end{center}
\caption{Nondeterministic automaton of the reverse of a right ideal.} 
\label{fig:RightIdealB}
\end{figure}

\begin{theorem}[Reverse of Right Ideal]
\label{thm:RightIdealRev}
The reverse of the right ideal $L(\cA'_n)$  has $2^{n-1}$ quotients.
\end{theorem}
\begin{proof}
Let $Z$ be the set of  words of the form 
$w=d(ae_{j-1})(ae_{j-2})\cdots (e_{1}a)(e_{0})$, where $0\le j\le n-2$, $e_i\in\{\eps,d\}$ for $1\le i\le j$.
In the subset construction applied to $\cN_n$, word $da^{j}$ reaches  $\{n-2-j,n-1\}$.
For $1\le i < j$, the set of states reached by $w$ includes state $n-2-i$ if and only if $e_i=d$. 
Thus each word $w$ reaches states $n-1$, $n-2-j$, and a different subset of $\{n-2, n-3,\ldots,n-2-(j-1)\}$. 
There are $2^j$ such subsets. 
As $j$ ranges from $0$ to $n-2$, we get $2^0+2^1+\cdots + 2^{n-2}$ different subsets.
Adding the subset $\{n-1\}$ reached by $\eps$, we get
 $2^{n-1}$ reachable subsets, each containing state $n-1$. 
 
Let $K=L^R$. The only state accepting $da^{n-2}$ is  $n-1$ reached by $\eps$. 
 If $S$ and $T$ are two different  subsets of $\{0,\ldots,n-1\}$ reachable by $u$ and $v$, respectively, and $i\in S\setminus T$,
 then $a^{n-2-i}\in K_u\setminus K_v$.
  Hence all the words in $\{\eps\}\cup Z$ are pairwise distinguishable, and $K=L^R$ has $2^{n-1}$
 distinct quotients.
\qed
\end{proof}

\section{Left Ideals and Suffix-Closed Languages}
\label{sec:leftideals12}

We provide strong support for the following conjecture about left ideals and suffix-closed languages:
\medskip

\noin
{\bf Conjecture 1 (Left Ideals and Suffix-Closed Languages).}
{\it If $L$ is a left ideal or a suffix-closed language with quotient complexity $\kappa(L)=n\ge 1$, then its syntactic complexity is less than or equal to $n^{n-1}+n-1$.}
\medskip

We show in this section that this complexity can be reached.
Since every suffix-closed language other than $\Sigma^*$ is the complement of a left ideal and complementation preserves syntactic complexity, it suffices to consider only left ideals.
Before attacking the conjecture itself, we prove some auxiliary results.

First we recall a result of Restivo and Vaglica~\cite{ReVa10}.
Consider a  semiautomaton $\cS=(P\cup\{0\},\Sig, \delta)$, where $0$ is a sink state, meaning that $\delta(0,a)=0$ for all $a\in \Sig$, and $P$ is strongly connected. Such a semiautomaton is \emph{uniformly minimal} if the automaton
$\cA=(P\cup \{0\}, \Sig,\delta,q_0, F)$ is minimal for every $q_0\in P$ and $\emp\not\subseteq F\subseteq P$.

One can test whether a semiautomaton is uniformly minimal with the aid of the directed \emph{pair graph} $G=G(\cS)=(V,E)$. The vertices of $G$ are all the unordered pairs $(p,q)$ of states with $p\neq q$.
There is an edge from $(p,q)$ to $(r,s)$ if and only if $\delta(p,a)=r$ and $\delta(q,a)=s$ for some $a\in \Sig$.
Then $\cS$ is uniformly minimal if and only if, for any pair $(p,q)$, there is a path to $(0,r)$ for some $r\in P$.

\begin{definition}
\label{def:semileft}
Let $n\ge 3$, and let $\cS_n$ be the  semiautomaton 
$$\cS_n =(\{0,\ldots,n-1\},\{a,b,c,d,e\},\delta),$$
where $a=(1,2,\ldots,n-1)$,
$b=(1,2)$,
$c={n-1\choose 1}$,
$d={n-1\choose 0}$,
and $e$ is the uniform transformation $Q \choose 1$.
The state graph of $\cS_n$ is shown in Fig.~\ref{fig:SemiautomatonS}.
For $n=3$ inputs $a$ and $b$ coincide; hence here we use $\Sig=\{b,c,d,e\}$.
\end{definition}

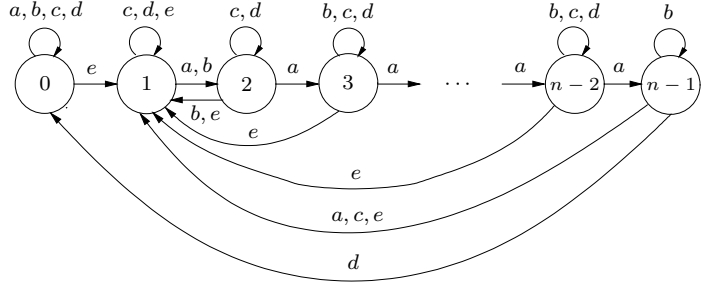
\begin{figure}[hbt]
\begin{center}
\setlength{\unitlength}{0.00048119in}
\begingroup\makeatletter\ifx\SetFigFont\undefined%
\gdef\SetFigFont#1#2#3#4#5{%
  \reset@font\fontsize{#1}{#2pt}%
  \fontfamily{#3}\fontseries{#4}\fontshape{#5}%
  \selectfont}%
\fi\endgroup%
{\renewcommand{\dashlinestretch}{30}
\begin{picture}(7430,3048)(0,-10)
\put(7028,2825){\makebox(0,0)[lb]{\smash{{\SetFigFont{8}{9.6}{\familydefault}{\mddefault}{\updefault}$b$}}}}
\put(6059.000,2575.333){\arc{333.333}{2.2143}{7.2105}}
\blacken\path(6196.638,2538.417)(6159.000,2442.000)(6237.107,2509.913)(6196.638,2538.417)
\put(2549.000,2575.333){\arc{333.333}{2.2143}{7.2105}}
\blacken\path(2686.638,2538.417)(2649.000,2442.000)(2727.107,2509.913)(2686.638,2538.417)
\put(1469.000,2575.333){\arc{333.333}{2.2143}{7.2105}}
\blacken\path(1606.638,2538.417)(1569.000,2442.000)(1647.107,2509.913)(1606.638,2538.417)
\put(7096.000,2553.333){\arc{333.333}{2.2143}{7.2105}}
\blacken\path(7233.638,2516.417)(7196.000,2420.000)(7274.107,2487.913)(7233.638,2516.417)
\put(3654.000,2560.333){\arc{333.333}{2.2143}{7.2105}}
\blacken\path(3791.638,2523.417)(3754.000,2427.000)(3832.107,2494.913)(3791.638,2523.417)
\put(393,2127){\ellipse{630}{630}}
\put(2557,2127){\ellipse{630}{630}}
\put(1487,2125){\ellipse{630}{630}}
\put(6077,2127){\ellipse{630}{630}}
\put(3649,2137){\ellipse{630}{630}}
\put(7107,2116){\ellipse{630}{630}}
\path(1794,2127)(2244,2127)
\blacken\path(2124.000,2097.000)(2244.000,2127.000)(2124.000,2157.000)(2124.000,2097.000)
\path(2874,2127)(3324,2127)
\blacken\path(3204.000,2097.000)(3324.000,2127.000)(3204.000,2157.000)(3204.000,2097.000)
\path(714,2127)(1164,2127)
\blacken\path(1044.000,2097.000)(1164.000,2127.000)(1044.000,2157.000)(1044.000,2097.000)
\path(6391,2127)(6789,2127)
\blacken\path(6669.000,2097.000)(6789.000,2127.000)(6669.000,2157.000)(6669.000,2097.000)
\path(5297,2127)(5747,2127)
\blacken\path(5627.000,2097.000)(5747.000,2127.000)(5627.000,2157.000)(5627.000,2097.000)
\path(3962,2127)(4412,2127)
\blacken\path(4292.000,2097.000)(4412.000,2127.000)(4292.000,2157.000)(4292.000,2097.000)
\path(2304,1955)(1749,1955)
\blacken\path(1869.000,1985.000)(1749.000,1955.000)(1869.000,1925.000)(1869.000,1985.000)
\path(3543,1833)(3542,1832)(3540,1831)
	(3536,1828)(3529,1824)(3520,1817)
	(3508,1809)(3493,1800)(3476,1788)
	(3455,1774)(3431,1759)(3405,1743)
	(3376,1725)(3346,1707)(3314,1688)
	(3280,1669)(3245,1650)(3209,1631)
	(3171,1612)(3133,1594)(3093,1576)
	(3052,1560)(3009,1544)(2965,1529)
	(2920,1515)(2872,1503)(2822,1492)
	(2770,1482)(2716,1475)(2660,1469)
	(2603,1466)(2545,1466)(2484,1469)
	(2425,1475)(2368,1483)(2314,1494)
	(2264,1507)(2217,1521)(2172,1537)
	(2131,1554)(2091,1573)(2054,1592)
	(2018,1612)(1984,1633)(1952,1654)
	(1921,1676)(1891,1697)(1863,1719)
	(1836,1740)(1812,1761)(1789,1780)
	(1768,1799)(1750,1815)(1735,1829)
	(1721,1842)(1711,1852)(1703,1859)(1691,1871)
\blacken\path(1797.066,1807.360)(1691.000,1871.000)(1754.640,1764.934)(1797.066,1807.360)
\path(5858,1894)(5857,1893)(5856,1892)
	(5853,1889)(5849,1884)(5843,1878)
	(5836,1869)(5826,1858)(5814,1845)
	(5800,1830)(5784,1814)(5766,1795)
	(5746,1775)(5725,1754)(5702,1732)
	(5677,1709)(5651,1685)(5623,1660)
	(5594,1636)(5563,1611)(5531,1585)
	(5496,1559)(5459,1533)(5420,1506)
	(5378,1479)(5333,1452)(5284,1423)
	(5232,1395)(5176,1366)(5116,1336)
	(5054,1307)(4988,1279)(4925,1253)
	(4863,1229)(4803,1207)(4746,1187)
	(4694,1169)(4646,1152)(4603,1138)
	(4566,1126)(4533,1115)(4505,1105)
	(4481,1097)(4461,1090)(4443,1084)
	(4428,1079)(4415,1075)(4403,1071)
	(4392,1067)(4380,1064)(4368,1061)
	(4354,1058)(4338,1054)(4320,1051)
	(4298,1047)(4272,1043)(4242,1039)
	(4206,1034)(4165,1029)(4118,1024)
	(4066,1019)(4007,1014)(3943,1010)
	(3875,1006)(3803,1003)(3729,1002)
	(3655,1002)(3583,1004)(3514,1007)
	(3450,1011)(3392,1015)(3339,1020)
	(3292,1025)(3251,1029)(3216,1034)
	(3185,1038)(3159,1042)(3137,1046)
	(3119,1049)(3103,1053)(3089,1057)
	(3077,1060)(3066,1064)(3054,1068)
	(3042,1072)(3029,1076)(3014,1082)
	(2997,1087)(2976,1094)(2952,1102)
	(2923,1111)(2890,1122)(2852,1133)
	(2809,1147)(2761,1162)(2707,1179)
	(2650,1198)(2589,1218)(2525,1241)
	(2461,1264)(2394,1290)(2329,1317)
	(2268,1343)(2211,1370)(2157,1396)
	(2107,1421)(2060,1446)(2016,1471)
	(1975,1495)(1937,1518)(1901,1542)
	(1867,1565)(1834,1587)(1804,1610)
	(1775,1631)(1747,1653)(1721,1674)
	(1696,1694)(1673,1713)(1652,1730)
	(1633,1747)(1616,1762)(1601,1775)
	(1589,1787)(1578,1797)(1570,1804)
	(1563,1810)(1554,1819)
\blacken\path(1660.066,1755.360)(1554.000,1819.000)(1617.640,1712.934)(1660.066,1755.360)
\path(624,1872)(631,1872)
\path(631,1872)(624,1872)
\path(6856,1903)(6855,1903)(6854,1902)
	(6852,1900)(6848,1897)(6842,1893)
	(6834,1887)(6824,1879)(6811,1870)
	(6796,1859)(6777,1846)(6756,1830)
	(6732,1813)(6705,1793)(6674,1772)
	(6641,1748)(6604,1723)(6565,1695)
	(6523,1666)(6478,1635)(6431,1603)
	(6381,1570)(6330,1535)(6276,1499)
	(6220,1463)(6163,1426)(6104,1388)
	(6043,1350)(5981,1312)(5918,1273)
	(5854,1235)(5789,1197)(5723,1159)
	(5656,1121)(5588,1084)(5519,1047)
	(5449,1011)(5378,976)(5307,941)
	(5234,907)(5160,875)(5085,843)
	(5008,811)(4930,782)(4851,753)
	(4771,725)(4689,699)(4605,674)
	(4519,650)(4432,628)(4344,608)
	(4253,590)(4162,574)(4069,560)
	(3975,548)(3881,539)(3786,533)
	(3692,530)(3588,530)(3487,534)
	(3387,541)(3291,551)(3199,565)
	(3109,581)(3024,599)(2942,620)
	(2863,643)(2787,668)(2715,695)
	(2646,723)(2579,753)(2515,785)
	(2454,817)(2395,851)(2338,887)
	(2283,923)(2230,960)(2179,998)
	(2129,1037)(2081,1076)(2035,1116)
	(1990,1157)(1946,1197)(1904,1238)
	(1863,1279)(1824,1319)(1786,1359)
	(1750,1399)(1715,1437)(1682,1475)
	(1651,1511)(1622,1545)(1595,1578)
	(1569,1609)(1546,1638)(1525,1665)
	(1506,1689)(1489,1711)(1474,1731)
	(1461,1748)(1451,1762)(1442,1774)
	(1434,1784)(1429,1791)(1425,1797)(1419,1805)
\blacken\path(1515.000,1727.000)(1419.000,1805.000)(1467.000,1691.000)(1515.000,1727.000)
\path(7111,1805)(7109,1803)(7107,1801)
	(7104,1798)(7100,1794)(7093,1787)
	(7085,1779)(7074,1769)(7062,1756)
	(7047,1741)(7029,1724)(7009,1705)
	(6986,1683)(6961,1658)(6933,1631)
	(6902,1602)(6869,1571)(6833,1537)
	(6794,1501)(6753,1464)(6710,1425)
	(6664,1384)(6617,1341)(6567,1298)
	(6516,1253)(6463,1208)(6408,1161)
	(6352,1115)(6294,1067)(6235,1020)
	(6175,972)(6113,925)(6051,878)
	(5987,831)(5922,784)(5856,738)
	(5789,693)(5720,648)(5651,604)
	(5580,561)(5508,519)(5434,478)
	(5359,438)(5283,399)(5205,361)
	(5125,325)(5043,290)(4959,256)
	(4873,224)(4785,194)(4695,165)
	(4603,138)(4509,114)(4412,92)
	(4313,72)(4212,54)(4109,39)
	(4004,28)(3898,19)(3791,14)
	(3684,12)(3573,14)(3464,20)
	(3355,29)(3249,42)(3144,58)
	(3042,77)(2943,98)(2846,122)
	(2752,148)(2661,177)(2572,207)
	(2485,239)(2402,273)(2320,309)
	(2240,346)(2163,384)(2088,424)
	(2014,465)(1942,507)(1872,551)
	(1804,595)(1737,640)(1671,687)
	(1606,734)(1543,781)(1481,829)
	(1421,878)(1361,927)(1303,976)
	(1246,1025)(1191,1074)(1137,1123)
	(1084,1172)(1033,1219)(984,1266)
	(936,1312)(891,1357)(847,1401)
	(805,1443)(766,1483)(729,1521)
	(694,1557)(662,1591)(632,1623)
	(605,1652)(580,1678)(558,1702)
	(539,1724)(522,1742)(507,1758)
	(495,1772)(484,1783)(476,1792)
	(470,1799)(466,1805)(459,1812)
\blacken\path(565.066,1748.360)(459.000,1812.000)(522.640,1705.934)(565.066,1748.360)
\put(338,2060){\makebox(0,0)[lb]{\smash{{\SetFigFont{8}{9.6}{\rmdefault}{\mddefault}{\updefault}$0$}}}}
\put(1426,2060){\makebox(0,0)[lb]{\smash{{\SetFigFont{8}{9.6}{\rmdefault}{\mddefault}{\updefault}$1$}}}}
\put(5812,2064){\makebox(0,0)[lb]{\smash{{\SetFigFont{7}{8.4}{\familydefault}{\mddefault}{\updefault}$n-2$}}}}
\put(2994,2247){\makebox(0,0)[lb]{\smash{{\SetFigFont{8}{9.6}{\familydefault}{\mddefault}{\updefault}$a$}}}}
\put(2498,2052){\makebox(0,0)[lb]{\smash{{\SetFigFont{8}{9.6}{\rmdefault}{\mddefault}{\updefault}$2$}}}}
\put(3586,2068){\makebox(0,0)[lb]{\smash{{\SetFigFont{8}{9.6}{\rmdefault}{\mddefault}{\updefault}$3$}}}}
\put(6846,2064){\makebox(0,0)[lb]{\smash{{\SetFigFont{7}{8.4}{\familydefault}{\mddefault}{\updefault}$n-1$}}}}
\put(15,2862){\makebox(0,0)[lb]{\smash{{\SetFigFont{8}{9.6}{\familydefault}{\mddefault}{\updefault}$a,b,c,d$}}}}
\put(2378,2847){\makebox(0,0)[lb]{\smash{{\SetFigFont{8}{9.6}{\familydefault}{\mddefault}{\updefault}$c,d$}}}}
\put(3371,2840){\makebox(0,0)[lb]{\smash{{\SetFigFont{8}{9.6}{\familydefault}{\mddefault}{\updefault}$b,c,d$}}}}
\put(5799,2841){\makebox(0,0)[lb]{\smash{{\SetFigFont{8}{9.6}{\familydefault}{\mddefault}{\updefault}$b,c,d$}}}}
\put(842,2232){\makebox(0,0)[lb]{\smash{{\SetFigFont{8}{9.6}{\familydefault}{\mddefault}{\updefault}$e$}}}}
\put(4067,2247){\makebox(0,0)[lb]{\smash{{\SetFigFont{8}{9.6}{\familydefault}{\mddefault}{\updefault}$a$}}}}
\put(5432,2261){\makebox(0,0)[lb]{\smash{{\SetFigFont{8}{9.6}{\familydefault}{\mddefault}{\updefault}$a$}}}}
\put(6489,2254){\makebox(0,0)[lb]{\smash{{\SetFigFont{8}{9.6}{\familydefault}{\mddefault}{\updefault}$a$}}}}
\put(1853,2255){\makebox(0,0)[lb]{\smash{{\SetFigFont{8}{9.6}{\familydefault}{\mddefault}{\updefault}$a,b$}}}}
\put(2574,1550){\makebox(0,0)[lb]{\smash{{\SetFigFont{8}{9.6}{\familydefault}{\mddefault}{\updefault}$e$}}}}
\put(1957,1746){\makebox(0,0)[lb]{\smash{{\SetFigFont{8}{9.6}{\familydefault}{\mddefault}{\updefault}$b,e$}}}}
\put(1231,2855){\makebox(0,0)[lb]{\smash{{\SetFigFont{8}{9.6}{\familydefault}{\mddefault}{\updefault}$c,d,e$}}}}
\put(3469,658){\makebox(0,0)[lb]{\smash{{\SetFigFont{8}{9.6}{\familydefault}{\mddefault}{\updefault}$a,c,e$}}}}
\put(4682,2068){\makebox(0,0)[lb]{\smash{{\SetFigFont{8}{9.6}{\familydefault}{\mddefault}{\updefault}$\cdots$}}}}
\put(3677,1102){\makebox(0,0)[lb]{\smash{{\SetFigFont{8}{9.6}{\familydefault}{\mddefault}{\updefault}$e$}}}}
\put(3640,141){\makebox(0,0)[lb]{\smash{{\SetFigFont{8}{9.6}{\familydefault}{\mddefault}{\updefault}$d$}}}}
\put(394.000,2560.333){\arc{333.333}{2.2143}{7.2105}}
\blacken\path(531.638,2523.417)(494.000,2427.000)(572.107,2494.913)(531.638,2523.417)
\end{picture}
}
\end{center}
\caption{Semiautomaton $\cS_n$  with $n^{n-1}+n-1$ transformations.} 
\label{fig:SemiautomatonS}
\end{figure}

\begin{definition}
Let $\Sig'=\Sig\setminus \{e\}$ and let $\cR_n$ be the semiautomaton
$\cR_n =(Q,\Sig',\delta'),$
where $Q=P\cup \{0\}$, $P=\{1,\ldots,n-1\}$, and $\delta'$ is 
the restriction of $\delta$ to $Q\times \Sig'$.
Note that 0 is a sink state of $\cR_n$.

\label{def:Blin}
\end{definition}

\begin{lemma}
\label{lem:BC}
The set $P$ is   strongly connected and $\cR_n$ is uniformly minimal.
\end{lemma}
\begin{proof}
Since $a$ is a cycle of the states in $P$,  $\cR_n$ is strongly connected. 

To show that $\cR_n$ is uniformly minimal, we construct the state-pair graph $G=G(\cR_n)$ of $\cR_n$, as in~\cite{ReVa10}. We need to show that for every vertex $v$ in $G$, there is a path to a vertex of the form $(0, j)$, where $j\in P$.

Assume that all the unordered pairs of distinct states of $\cR_n$ are represented as $(i,j)$, where $i<j$.
If a vertex  is of the form $(0,j)$, then there is nothing to prove.
If a vertex is of the form $(i,j)$, $0<i<j$, then applying $a^{n-1-j}$ reaches $(i+n-1-j, n-1)$.
Then $d$ takes the pair $(i+n-1-j, n-1)$ to $(0,i+n-1-j)$.
Consequently, $\cR_n$ is uniformly minimal.
\qed
\end{proof}

\begin{theorem}[Left Ideals and Suffix-Closed Languages]
\label{thm:LeftIdeal4}
For $n\ge 3$, let
$\cA_n=(Q, \Sig,\delta,0, F),$ where $(Q,\Sig,\delta)=\cS_n$ of Def.~\ref{def:semileft}, and $F$ is any non-empty subset of 
$Q\setminus\{0\}$.
Then $\cA_n$ is minimal, and the language $L=L(\cA_n)$ accepted by $\cA_n$ is a left ideal and has syntactic complexity $\sig(L)=n^{n-1}+n-1$.
\end{theorem}
\begin{proof}
Since semiautomaton $\cR_n$ is uniformly minimal, automaton $\cA_n$ is minimal for every choice of $F$.
Hence $L$ has $n$ quotients.

To prove that $L$ is a left ideal it suffices to show that, for any $w\in L$, we also have $hw\in L$ for every $h\in\Sig$. This is obvious if $h\in\Sig\setminus \{e\}$, since all transitions from state 0 under $h$ lead to state 0. If $w\in L$, then $w$ has the form $w =uev$, where $\delta(0,u)=0$,
$\delta(0,ue)=1$, and $v\in L_e$. 
But  $\delta(0,eue)=1$, since $\delta(i,eue)=1$ for all $i\in Q$, and $v\in L_e$ gives us $euev=ew\in L$.
Thus $L$ is a left ideal.
\medskip

Consider any transformation $t$ of the form
$$
t=\left( \begin{array}{cccccccc}
0 & 1 & 2 & 3 & \cdots & n-3 & n-2 & n-1 \\
0 & i_1 & i_2 & i_3 & \cdots & i_{n-3} & i_{n-2} & i_{n-1}
\end{array} \right ),
$$
where $i_k\in\{0,1,2,\ldots,n-2, n-1\}$ for $1\le k\le n-1$;
there are $n^{n-1}$ such transformations.
We have two cases:
\be
\item If $i_k\ne 0$ for all $k$, $1\le k\le n-1$, then all the images of the last $n-1$ states are in the set $\{1,\ldots,n-1\}$.
By Theorem~\ref{thm:salomaa}, $t$ can be performed by $\cA_n$.
\item If $i_h = 0$ for some $h$, $1\le h\le n-1$, then there exists some $j$,
$1\le j\le n-1$ such that $i_k\ne j$ for all $k$, $1\le k\le n-1$.
\ee
Define $i'_k$ for all $1\le k\le n-1$ as follows:
\begin{equation*} 
i'_k = \left\{ 
\begin{array}{ccc} 
j,	& & \quad \text{if } i_k=0;\\ 
i_k,   &   & \quad \text{if }  i_k\ne 0.\\ 
\end{array} \right. 
\end{equation*}
Let 
$$
s=\left( \begin{array}{cccccccc}
0 & 1 & 2 & 3 & \cdots & n-3 & n-2 & n-1 \\
0 & i'_1 & i'_2 & i'_3 & \cdots & i'_{n-3} & i'_{n-2} & i'_{n-1} 
\end{array} \right ),$$
and
$r=(j,n-1)$.
By Theorem~\ref{thm:salomaa}, $s$ and $r$ can be performed by $\cA_n$.

Now consider $srdr$.
If $t$ maps $k$ to 0, then $s$ maps $k$ to $j$,  $r$ maps $j$ to $n-1$, 
$d$ maps $n-1$ to 0, and $r$ maps 0 to 0.
If $t$ maps $k$ to $n-1$, then $s$ maps $k$ to $n-1$, $r$ maps $n-1$ to $j$, $d$ maps $j$ to $j$, and $r$ maps $j$ to $n-1$.
Finally, if $t$ maps $k$ to an element other than 0 or $n-1$, then $srdr$ maps $k$ to the same element.
Hence we have $t=srdr$, and $t$ can be performed by $\cA_n$ as well.
\medskip

Now consider any transformation $t$ that maps all the states to some state $j\ne 0$;
there are  $n-1$ such transformations.
We have two cases:
\be
\item If $j=1$, then $t=e$; therefore $t$ can be performed by $\cA_n$.
\item Otherwise, let $s=(1,j)$.
By Theorem~\ref{thm:salomaa}, $s$ can be performed by $\cA_n$.
Since $t=es$, $t$ can also be performed by $\cA_n$ as well.
\ee
In summary, the syntactic complexity of $L(\cA_n)$ is $n^{n-1}+n-1$.\qed
\end{proof}

Since inputs $a$ and $b$ of automaton $\cA_3$ coincide, we omit $a$. Table~\ref{tab:li3} shows the transition table of $\cA_3$ and  its $3^2+2=11$ transformations.
We will show that 11 is indeed the maximal bound for $n=3$, but we require more properties of left ideals.

\begin{table}[ht]
\caption{The eleven transformations of automaton $\cA_{3}$ of a left ideal.}
\label{tab:li3}
\begin{center}
$
\begin{array}{| c ||c|c| c| c||c|c|c|c|c|c|c|}    
\hline
\ \  \ \ &\ \ b \ \ &\ \ c \ \ &\ \ d \  \ & \  \ e \ \ 
&  \ bb \ &\ bd\ &\ cb \ &\ db \ 
&\  eb \  & \  bdb \ & \  cbd \
  \\
\hline  \hline
  0
& 0&0  &0&  1  	          &0 &0   & 0 & 0   	& 2 &0& 0\\
\hline
1 
& 2& 1 &1 & 1  		&1 & 0   & 2 & 2 	& 2 &0& 0\\
\hline
2
& 1 &1  & 0  & 1 	&2 & 1   & 2 &  0  	&2  &2&0\\
\hline
\hline
\end{array}
$
\end{center}
\end{table}

Let $\cA=(Q, \Sig,\delta,q_0, F)$ be the quotient automaton  
 of a left ideal.
For every word $w\in\Sig^*$, consider the sequence $q_0=p_0, p_1, p_2\ldots $ of states obtained by applying powers of $w$ to the initial state $q_0$, that is, let $p_i=\delta(q_0,w^i)$. Since $\cA$ has $n$ states, we must eventually have a repeated state in that sequence, that is, we must have 
some $i$ and $j>i$ such that $p_0,p_1,\ldots, p_i,p_{i+1},\ldots p_{j-1}$ are distinct and $p_j=p_i$.
The sequence $q_0=p_0,p_1,\ldots, p_i,p_{i+1},\ldots p_{j-1}$ of states with $p_j=p_i$ is called the \emph{behavior of $w$ on $\cA$}, and the integer $j-i$ is the \emph{period} of that behavior.
We will use the notation $\langle p_0,p_1,\ldots, p_i,p_{i+1},\ldots p_{j-1}; p_j=p_i\rangle$ for such behaviors.
If the period of $w$ is 1, then its behavior is \emph{aperiodic}; otherwise, it is \emph{periodic}.

\begin{lemma}
\label{lem:cycles}
If $\cA$ is the quotient automaton of a left ideal $L$, then the behavior of every word $w\in \Sig^*$ is aperiodic. 
Moreover, $L$ does not have the empty quotient.
\end{lemma}

\begin{proof}
Suppose that $w$ has the behavior 
$\langle q_0=p_0,p_1,\ldots, p_i,p_{i+1},\ldots p_{j-1};p_j=p_i\rangle$, where $j-i\ge 2$; then $j-1\ge i+1$. Since $\cA$ is minimal, states $p_i$ and $p_{j-1}$ must be distinguishable, say by word $x\in\Sig^*$.
If $w^ix\in L$, then $w^{j-1} x=w^iw^{j-i-1}x=w^{j-i-1}(w^ix) \not\in L$, contradicting the assumption that $L$ is a left ideal.
If $w^{j-1}x\in L$, then $w^jx=w(w^{j-1}x) \not\in L$, again contradicting that $L$ is a  left ideal.

For the second claim, we know that a left ideal is non-empty by definition. So suppose that $w\in L$. If $L$ has the empty quotient, say $L_x=\emp$, then $xw\not\in L$, which is a contradiction.
\qed
\end{proof}
\begin{example}

Note that the conditions of Lemma~\ref{lem:cycles} are not sufficient. For $\Sig=\{a,b\}$, the language $L=b\cup \Sig^*a$ satisfies the conditions, but is not a left ideal because $b\in L$ but $ab\not\in L$.
Its quotient automaton is shown in Fig.~\ref{fig:notideal}.

If the accepting state is 2 instead of 1, the language becomes $L'=\Sig\Sig^*b=\Sig^*\Sig b$, which \emph{is} a left ideal.
The languages $L$ and $L'$  have the same syntactic semigroup, but one is a left ideal while the other is not.
\end{example}

\begin{figure}[hbt]
\begin{center}
\setlength{\unitlength}{0.00048119in}
\begingroup\makeatletter\ifx\SetFigFont\undefined%
\gdef\SetFigFont#1#2#3#4#5{%
  \reset@font\fontsize{#1}{#2pt}%
  \fontfamily{#3}\fontseries{#4}\fontshape{#5}%
  \selectfont}%
\fi\endgroup%
{\renewcommand{\dashlinestretch}{30}
\begin{picture}(4045,1377)(0,-10)
\put(2884,734){\makebox(0,0)[lb]{\smash{{\SetFigFont{8}{9.6}{\familydefault}{\mddefault}{\updefault}$b$}}}}
\put(3713.000,899.000){\arc{298.530}{2.1848}{7.2400}}
\blacken\path(3540.118,865.042)(3627.000,777.000)(3591.771,895.570)(3540.118,865.042)
\put(2230,445){\ellipse{768}{768}}
\put(592,459){\ellipse{656}{656}}
\put(3709,468){\ellipse{656}{656}}
\put(2237,450){\ellipse{656}{656}}
\path(2554,636)(3431,636)
\blacken\path(3311.000,606.000)(3431.000,636.000)(3311.000,666.000)(3311.000,606.000)
\path(3440,276)(2563,276)
\blacken\path(2683.000,306.000)(2563.000,276.000)(2683.000,246.000)(2683.000,306.000)
\path(912,449)(1857,449)
\blacken\path(1737.000,419.000)(1857.000,449.000)(1737.000,479.000)(1737.000,419.000)
\path(12,456)(275,456)
\blacken\path(155.000,426.000)(275.000,456.000)(155.000,486.000)(155.000,426.000)
\put(1212,591){\makebox(0,0)[lb]{\smash{{\SetFigFont{8}{9.6}{\familydefault}{\mddefault}{\updefault}$a,b$}}}}
\put(3642,1131){\makebox(0,0)[lb]{\smash{{\SetFigFont{8}{9.6}{\familydefault}{\mddefault}{\updefault}$b$}}}}
\put(2187,381){\makebox(0,0)[lb]{\smash{{\SetFigFont{8}{9.6}{\familydefault}{\mddefault}{\updefault}$1$}}}}
\put(522,396){\makebox(0,0)[lb]{\smash{{\SetFigFont{8}{9.6}{\familydefault}{\mddefault}{\updefault}$0$}}}}
\put(3612,411){\makebox(0,0)[lb]{\smash{{\SetFigFont{8}{9.6}{\familydefault}{\mddefault}{\updefault}$2$}}}}
\put(2142,1191){\makebox(0,0)[lb]{\smash{{\SetFigFont{8}{9.6}{\familydefault}{\mddefault}{\updefault}$a$}}}}
\put(2967,73){\makebox(0,0)[lb]{\smash{{\SetFigFont{8}{9.6}{\familydefault}{\mddefault}{\updefault}$a$}}}}
\put(2213.000,953.000){\arc{298.530}{2.1848}{7.2400}}
\blacken\path(2040.118,919.042)(2127.000,831.000)(2091.771,949.570)(2040.118,919.042)
\end{picture}
}
\end{center}
\caption{Automaton of a language that is not a left ideal.} 
\label{fig:notideal}
\end{figure}
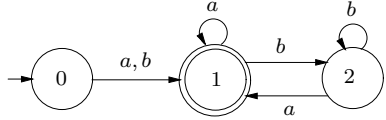

\begin{proposition}
\label{prop:cycles}
The number of transformations  ruled out by Lemma~\ref{lem:cycles} is
\begin{equation}
\label{eq:cycles}
 \sum_{j=2}^n {{n-1} \choose {j-1}}\, (j-1)!  (j-1) \,n^{n-j}
 =\sum_{j=2}^n\frac{(n-1)! } {(n-j)!}\,   (j-1) \,n^{n-j}. 
\end{equation}
\end{proposition}
\begin{proof}
Consider a behavior $\langle p_0,p_1,\ldots, p_i,p_{i+1},\ldots p_{j-1}; p_j=p_i\rangle$ of length $j$. 
The first state, $p_0$, must be $0$, but the set $\{p_1,\ldots p_{j-1}\}$ can be any subset of cardinality $j-1$ of the remaining $n-1$ states, and there are ${n-1}\choose {j-1}$ such subsets.
The states in each subset can be arranged in any order, giving $(j-1)!$ permutations. Then  there are $j-1$ choices for $p_j$.  Finally, $n-j$ states that are not part of the behavior can have $n$ transformations each, adding the factor $n^{n-j}$. 
\qed
\end{proof}

Lemma~\ref{lem:cycles} provides an upper bound to the syntactic complexity of left ideals, as shown in Table~\ref{tab:ruledout}. However, there is a large gap between this bound and the bound we can achieve, and we know that this bound cannot be reached for $n=3$.
\begin{table}[ht]
\caption{Number of transformations ruled out  by Lemma~\ref{lem:cycles}.}
\label{tab:ruledout}
\begin{center}
$
\begin{array}{| c |c|c| c| c|c|}    
\hline
\ \ n \ \ &\ \ 2 \ \ &\ \ 3 \ \ &\ \ 4 \  \ & \  \  5 \ \ 
 \  &\ \ \  \ldots \  \    \\
\hline  
n^n  & 4&27  &256 &\ 3,125 \ 	 & \  \ldots 	 \ \\
\hline  
\text{ruled out by lemma} & 1&10  &162 &\ 1,556 \ 	 & \  \ldots 	 \ \\
\hline  
\text{an upper bound} & 3&17  &94 &\ 1,569 \ 	 & \  \ldots 	 \ \\
\hline  
(n-1)^{n-1}+n-1& 3& 11  &67 &\ 629 \ 	 & \  \ldots 	 \ \\
\hline
\end{array}
$
\end{center}
\end{table}

\begin{theorem}[Small Left Ideals and Suffix-Closed Languages]
\label{thm:li3}
If $1\le n\le 3$ and $L$ is a left ideal or a suffix-closed language with $\kappa(L)=n$, then $\sig(L)\le n^{n-1}+n-1$. 
Moreover, the bound is tight for $n=1$ if $|\Sig|\ge 1$, 
for $n=2$  if  $|\Sig|\ge 3$, and for $n=3$ if  $|\Sig|\ge 4$.
\end{theorem}
\begin{proof}
We consider the three values of $n$ separately. The bounds claimed below have all been verified by a computer program.
\bd
\item[n=1:]
Here, there is only one type of left ideal,  $L=\Sig^*$.
Thus the  bound 1 holds, and  is met by $a^*$ over $\Sig=\{a\}$.
\medskip

\item[n=2:]
There is only one periodic behavior $\langle p_0=0,p_1=1; p_2=p_0\rangle$; hence only   transformation $[1,0]$ is ruled out by Lemma~\ref{lem:cycles}. Thus  the bound 3 holds.
\smallskip

Now consider any left ideal $L$ with $n=2$.
State 1 must be reachable from state 0, say by input~$a$. By Lemma~\ref{lem:cycles}, we cannot have $a:[1,0]$, and so we have $a:[1,1]$.
\medskip
	
	If $\Sig=\{a\}$, then we have the left ideal $L=a^*a$ with $\sig(L)=1$.\\
	Thus $\sig(L)=1$ if $|\Sig|=1$.
	\medskip
	
	If $\Sig=\{a,b\}$, then we have three cases:
		\be
		\item
		If $b:[1,1]$,  then $L=\Sig^*\Sig$ with $\sig(L)=1$.
		\item
		If  $b:[0,0]$,  then  $L=\Sig^*a$ with $\sig(L)=2$.
		\item
		If  $b:[0,1]$, then  $L=\Sig^*a\Sig^*$ with $\sig(L)=2$.
		\ee
Thus $\sig(L)\le2$ if $|\Sig|=2$.
\medskip
	
 If $\Sig=\{a,b,c\}$,  the language $L=\Sig^*a(a+b)^*$ meets the bound 3.
 \medskip
 
\item[n=3:]
For $|\Sig|=1$, there is only one left ideal,  namely $L=\Sig^*aa$, and it has $\sig(L)=2$.
\medskip

For $|\Sig|=2$, we have verified that the number of transformations is  at most 7, and the automaton with inputs $a:[001]$ and $b:[122]$ meets this bound.
\medskip

For $|\Sig|=3$, we have verified that the number of transformations is  at most 9, and the automaton $\cA_3$ of Theorem~\ref{thm:LeftIdeal4} restricted to inputs $b:[0,2,1]$, $d:[0,1,0]$ and $e:[1,1,1]$ meets this bound.
\medskip

Now consider the case $|\Sig|=4$.
For $n=3$, there are three types of periodic behaviors:
$(p_0,p_1; p_2=p_0)$,
$(p_0,p_1,p_2;p_3=p_0)$, and
$(p_0,p_1,p_2;p_3=p_1)$.
The following ten transformations are ruled out by Lemma~\ref{lem:cycles}:
$[1,0,0]$, $[1,0,1]$, $[1,0,2]$, $[1,2,0]$, $[1,2,1]$,
$[2,0,0]$, $[2,1,0]$, $[2,2,0]$, $[2,0,1]$, and $[2,2,1]$.

There are six transformations that are not ruled out  by Lemma~\ref{lem:cycles} and that do not appear in Table~\ref{tab:li3}, namely: $[1,1,0]$, $[1,1,2]$, $[1,2,2]$, $[2,0,2]$, $[2,1,1]$,  and 
$[2,1,2]$.
Each of these transformations, when followed by a transformation from Table~\ref{tab:li3}, results in a transformation ruled out  by Lemma~\ref{lem:cycles}:\\
$t_1:[1,1,0]$ and $cb:[0,2,2]$ yield $t_1cb:[2,2,0]$,\\
$t_2:[1,1,2]$ and $db:[0,2,0]$ yield $t_2db:[2,2,0]$,\\
$t_3:[1,2,2]$ and $d:[0,1,0]$ yield $t_3d:[1,0,0]$,\\
$t_4:[2,0,2]$ and $c:[0,1,1]$ yield $t_4c:[1,0,1]$,\\
$t_5:[2,1,1]$ and $bdb:[0,0,2]$ yield $t_5bdb:[2,0,0]$,\\
$t_6:[2,1,2]$ and $bd:[0,0,1]$ yield $t_6bd:[1,0,1]$.\\
All these conflicts are independent of the set of accepting states. Furthermore, each transformation  not ruled out by  Lemma~\ref{lem:cycles} conflicts with a \emph{different} transformation from Table~\ref{tab:li3}. So at most one transformation can be chosen from each pair, showing that there cannot be more than 11 transformations for any automaton with three states.
Hence the syntactic complexity of any left ideal with quotient complexity 3 is at most 11, and the example of Table~\ref{tab:li3} shows that this bound is tight.	\qed
\ed

\end{proof}

Table~\ref{tab:LISummary} summarizes our results concerning left ideals. 
The figures in bold type are tight upper bounds. The other  complexities are achievable, but we have no proof that they are upper bounds. In general, there are many solutions with the same complexity.

The complexity 17 for $n=4$, $|\Sig|=2$ is reached with the inputs $a:[1,2,3,3]$ and $b:[0,0,1,2]$.
The complexity 25 for $n=4$, $|\Sig|=3$ is met by $\cA_4$ of Theorem~\ref{thm:LeftIdeal4}
restricted to $a,d,e$. 
The complexity 64 for $n=4$, $|\Sig|=4$ is met by $\cA_4$ of Theorem~\ref{thm:LeftIdeal4}
restricted to $a,c,d,e$. 

The complexity 34 for $n=4$, $|\Sig|=2$ is reached with the inputs $a:[1,2,3,4,4]$ and $b:[0,0,1,2,3]$.
The complexity 65 for $n=5$, $|\Sig|=3$ is met by $\cA_5$ of Theorem~\ref{thm:LeftIdeal4}
restricted to $a,d,e$. 
The complexity of 453 for $n=5$, $|\Sig|=4$ is met by $\cA_5$ of Theorem~\ref{thm:LeftIdeal4}
restricted to $a,c,d,e$. 
\medskip

\begin{table}[ht]
\caption{Syntactic complexities for left ideals.}
\label{tab:LISummary}
\begin{center}
$
\begin{array}{| c ||c|c| c| c|c|c|c|}    
\hline
\ \  \ \ &\ \ n=1 \ \ &\ \ n=2 \ \ &\ \ n=3 \  \ & \  \  n=4 \ \ 
&  \ n=5 \  &\ \ \ \ 
&\  n=n \  
  \\
\hline  \hline
  |\Sig|=1
& \bf 1&\bf 1  &\bf 2&\bf  3  	&\bf 4    &   \ldots 	& \bf n-1 \\
\hline
|\Sig|=2 
& - & \bf 2 &\bf 7 &   	17	& 34  &     \ldots&  \\
\hline
|\Sig|=3
& - &\bf 3 &\bf 9  & 25 	& 65  &   \ldots  	& \\
\hline
|\Sig|=4
& - &-  & \bf 11 & 64 	& 453 &      \ldots	& \\\hline
|\Sig|=5
& - &-  & -        & 67 	& 629 &   \ldots	& n^{n-1} +n-1 \\
\hline
\end{array}
$
\end{center}
\end{table}

As was the case with right ideals, for our left ideal with maximal syntactic complexity, the reverse language has maximal state complexity, as the next result shows. 
This time, however, we require an alphabet of four letters.

\begin{theorem}[Reverse of Left Ideal]
The reverse of the left ideal accepted by automaton $\cA_n$ of Theorem~\ref{thm:LeftIdeal4} restricted to $\{a,c,d,e\}$
has $2^{n-1}+1$ quotients, which is the maximum possible for a left ideal.
\end{theorem}
\begin{proof}
Consider the subset construction applied to the nondeterministic automaton of Fig.~\ref{fig:LeftReverse}.
First we show that  the subset $Q=\{0,1,\ldots,n-1\}$ and all subsets of $P=Q\setminus\{0\}$ are reachable. 
The word $a^{n-2}e$ reaches $Q$, and 
the word $(a^{n-2}c)^{n-2}$ reaches $P$.
Now suppose we have a set $S$ of $k$ elements,
$S=\{i_1,i_2,\ldots,i_k\}$, where $\{1\le i_1<i_2<\cdots < i_k\le n-1\}$. To delete the $j$th element of $S$ apply $a^{i_j} da^{n-1-i{_j}}$.
Hence all subsets of $P$ can be reached.

Note now that $a^{i-1}e$ is accepted only from state $i$, for $i=1,\ldots,n-1$, and
the empty word is accepted only from state 0.
It follows that all subsets of $P$ are pairwise distinguishable. 
\qed
\end{proof}

\begin{figure}[h]
\begin{center}
\setlength{\unitlength}{0.00052493in}
\begingroup\makeatletter\ifx\SetFigFont\undefined%
\gdef\SetFigFont#1#2#3#4#5{%
  \reset@font\fontsize{#1}{#2pt}%
  \fontfamily{#3}\fontseries{#4}\fontshape{#5}%
  \selectfont}%
\fi\endgroup%
{\renewcommand{\dashlinestretch}{30}
\begin{picture}(7724,2973)(0,-10)
\put(5912,2706){\makebox(0,0)[lb]{\smash{{\SetFigFont{8}{9.6}{\familydefault}{\mddefault}{\updefault}$c,d$}}}}
\put(1485.000,2448.333){\arc{333.333}{2.2143}{7.2105}}
\blacken\path(1622.638,2411.417)(1585.000,2315.000)(1663.107,2382.913)(1622.638,2411.417)
\put(3670.000,2433.333){\arc{333.333}{2.2143}{7.2105}}
\blacken\path(3807.638,2396.417)(3770.000,2300.000)(3848.107,2367.913)(3807.638,2396.417)
\put(6077.000,2433.333){\arc{333.333}{2.2143}{7.2105}}
\blacken\path(6214.638,2396.417)(6177.000,2300.000)(6255.107,2367.913)(6214.638,2396.417)
\put(367.000,2471.333){\arc{333.333}{2.2143}{7.2105}}
\blacken\path(504.638,2434.417)(467.000,2338.000)(545.107,2405.913)(504.638,2434.417)
\put(2573,2000){\ellipse{630}{630}}
\put(1503,1998){\ellipse{630}{630}}
\put(6093,2000){\ellipse{630}{630}}
\put(3665,2010){\ellipse{630}{630}}
\put(7123,1989){\ellipse{630}{630}}
\put(368,1990){\ellipse{720}{720}}
\put(367,1992){\ellipse{630}{630}}
\blacken\path(1930.000,2030.000)(1810.000,2000.000)(1930.000,1970.000)(1930.000,2030.000)
\path(1810,2000)(2260,2000)
\blacken\path(3010.000,2030.000)(2890.000,2000.000)(3010.000,1970.000)(3010.000,2030.000)
\path(2890,2000)(3340,2000)
\blacken\path(850.000,2030.000)(730.000,2000.000)(850.000,1970.000)(850.000,2030.000)
\path(730,2000)(1180,2000)
\blacken\path(6527.000,2030.000)(6407.000,2000.000)(6527.000,1970.000)(6527.000,2030.000)
\path(6407,2000)(6805,2000)
\blacken\path(5433.000,2030.000)(5313.000,2000.000)(5433.000,1970.000)(5433.000,2030.000)
\path(5313,2000)(5763,2000)
\blacken\path(4098.000,2030.000)(3978.000,2000.000)(4098.000,1970.000)(4098.000,2030.000)
\path(3978,2000)(4428,2000)
\blacken\path(2200.000,1798.000)(2320.000,1828.000)(2200.000,1858.000)(2200.000,1798.000)
\path(2320,1828)(1765,1828)
\path(7712,1992)(7442,1992)
\blacken\path(7562.000,2022.000)(7442.000,1992.000)(7562.000,1962.000)(7562.000,2022.000)
\blacken\path(3474.281,1615.874)(3559.000,1706.000)(3441.836,1666.344)(3474.281,1615.874)
\path(3559,1706)(3545,1697)(3536,1690)
	(3524,1682)(3509,1673)(3492,1661)
	(3471,1647)(3447,1632)(3421,1616)
	(3392,1598)(3362,1580)(3330,1561)
	(3296,1542)(3261,1523)(3225,1504)
	(3187,1485)(3149,1467)(3109,1449)
	(3068,1433)(3025,1417)(2981,1402)
	(2936,1388)(2888,1376)(2838,1365)
	(2786,1355)(2732,1348)(2676,1342)
	(2619,1339)(2561,1339)(2500,1342)
	(2441,1348)(2384,1356)(2330,1367)
	(2280,1380)(2233,1394)(2188,1410)
	(2147,1427)(2107,1446)(2070,1465)
	(2034,1485)(2000,1506)(1968,1527)
	(1937,1549)(1907,1570)(1879,1592)
	(1852,1613)(1828,1634)(1805,1653)
	(1784,1672)(1766,1688)(1751,1702)
	(1737,1715)(1727,1725)(1719,1732)
	(1713,1738)(1710,1741)(1708,1743)(1707,1744)
\blacken\path(5816.023,1657.736)(5874.000,1767.000)(5771.425,1697.874)(5816.023,1657.736)
\path(5874,1767)(5865,1757)(5859,1751)
	(5852,1742)(5842,1731)(5830,1718)
	(5816,1703)(5800,1687)(5782,1668)
	(5762,1648)(5741,1627)(5718,1605)
	(5693,1582)(5667,1558)(5639,1533)
	(5610,1509)(5579,1484)(5547,1458)
	(5512,1432)(5475,1406)(5436,1379)
	(5394,1352)(5349,1325)(5300,1296)
	(5248,1268)(5192,1239)(5132,1209)
	(5070,1180)(5004,1152)(4941,1126)
	(4879,1102)(4819,1080)(4762,1060)
	(4710,1042)(4662,1025)(4619,1011)
	(4582,999)(4549,988)(4521,978)
	(4497,970)(4477,963)(4459,957)
	(4444,952)(4431,948)(4419,944)
	(4408,940)(4396,937)(4384,934)
	(4370,931)(4354,927)(4336,924)
	(4314,920)(4288,916)(4258,912)
	(4222,907)(4181,902)(4134,897)
	(4082,892)(4023,887)(3959,883)
	(3891,879)(3819,876)(3745,875)
	(3671,875)(3599,877)(3530,880)
	(3466,884)(3408,888)(3355,893)
	(3308,898)(3267,902)(3232,907)
	(3201,911)(3175,915)(3153,919)
	(3135,922)(3119,926)(3105,930)
	(3093,933)(3082,937)(3070,941)
	(3058,945)(3045,949)(3030,955)
	(3013,960)(2992,967)(2968,975)
	(2939,984)(2906,995)(2868,1006)
	(2825,1020)(2777,1035)(2723,1052)
	(2666,1071)(2605,1091)(2541,1114)
	(2477,1137)(2410,1163)(2345,1190)
	(2284,1216)(2227,1243)(2173,1269)
	(2123,1294)(2076,1319)(2032,1344)
	(1991,1368)(1953,1391)(1917,1415)
	(1883,1438)(1850,1460)(1820,1483)
	(1791,1504)(1763,1526)(1737,1547)
	(1712,1567)(1689,1586)(1668,1603)
	(1649,1620)(1632,1635)(1617,1648)
	(1605,1660)(1594,1670)(1586,1677)
	(1579,1683)(1575,1687)(1572,1690)
	(1571,1691)(1570,1692)
\path(640,1745)(647,1745)
\path(647,1745)(640,1745)
\blacken\path(6794.000,1680.000)(6872.000,1776.000)(6758.000,1728.000)(6794.000,1680.000)
\path(6872,1776)(6864,1770)(6858,1766)
	(6850,1760)(6840,1752)(6827,1743)
	(6812,1732)(6793,1719)(6772,1703)
	(6748,1686)(6721,1666)(6690,1645)
	(6657,1621)(6620,1596)(6581,1568)
	(6539,1539)(6494,1508)(6447,1476)
	(6397,1443)(6346,1408)(6292,1372)
	(6236,1336)(6179,1299)(6120,1261)
	(6059,1223)(5997,1185)(5934,1146)
	(5870,1108)(5805,1070)(5739,1032)
	(5672,994)(5604,957)(5535,920)
	(5465,884)(5394,849)(5323,814)
	(5250,780)(5176,748)(5101,716)
	(5024,684)(4946,655)(4867,626)
	(4787,598)(4705,572)(4621,547)
	(4535,523)(4448,501)(4360,481)
	(4269,463)(4178,447)(4085,433)
	(3991,421)(3897,412)(3802,406)
	(3708,403)(3604,403)(3503,407)
	(3403,414)(3307,424)(3215,438)
	(3125,454)(3040,472)(2958,493)
	(2879,516)(2803,541)(2731,568)
	(2662,596)(2595,626)(2531,658)
	(2470,690)(2411,724)(2354,760)
	(2299,796)(2246,833)(2195,871)
	(2145,910)(2097,949)(2051,989)
	(2006,1030)(1962,1070)(1920,1111)
	(1879,1152)(1840,1192)(1802,1232)
	(1766,1272)(1731,1310)(1698,1348)
	(1667,1384)(1638,1418)(1611,1451)
	(1585,1482)(1562,1511)(1541,1538)
	(1522,1562)(1505,1584)(1490,1604)
	(1477,1621)(1467,1635)(1458,1647)
	(1450,1657)(1445,1664)(1441,1670)
	(1438,1674)(1436,1676)(1435,1677)(1435,1678)
\blacken\path(7011.360,1556.934)(7075.000,1663.000)(6968.934,1599.360)(7011.360,1556.934)
\path(7075,1663)(7063,1651)(7057,1645)
	(7049,1637)(7038,1626)(7025,1614)
	(7010,1599)(6993,1582)(6973,1562)
	(6950,1540)(6924,1515)(6896,1489)
	(6866,1460)(6832,1428)(6797,1395)
	(6759,1360)(6718,1323)(6675,1284)
	(6630,1244)(6584,1203)(6535,1161)
	(6484,1117)(6432,1073)(6378,1029)
	(6323,984)(6266,939)(6208,893)
	(6149,848)(6089,803)(6027,758)
	(5964,714)(5900,670)(5835,627)
	(5768,585)(5700,543)(5631,503)
	(5560,463)(5488,424)(5414,386)
	(5339,350)(5261,314)(5182,281)
	(5100,248)(5017,217)(4931,188)
	(4842,160)(4752,134)(4658,110)
	(4563,88)(4465,69)(4364,52)
	(4261,38)(4157,27)(4051,18)
	(3943,13)(3835,12)(3727,14)
	(3619,20)(3513,29)(3408,41)
	(3305,56)(3205,73)(3107,93)
	(3011,116)(2917,140)(2826,167)
	(2738,195)(2651,225)(2567,257)
	(2486,290)(2406,324)(2328,360)
	(2252,397)(2178,435)(2105,474)
	(2034,515)(1965,556)(1896,598)
	(1830,641)(1764,684)(1699,729)
	(1636,773)(1574,818)(1513,864)
	(1454,909)(1395,955)(1338,1001)
	(1282,1046)(1228,1091)(1176,1135)
	(1125,1179)(1076,1222)(1028,1263)
	(983,1304)(940,1343)(899,1380)
	(861,1415)(825,1449)(791,1480)
	(760,1510)(732,1536)(706,1561)
	(683,1583)(663,1603)(645,1620)
	(630,1635)(617,1648)(607,1658)
	(598,1667)(592,1673)(587,1678)
	(584,1681)(582,1683)(580,1685)
\put(1442,1933){\makebox(0,0)[lb]{\smash{{\SetFigFont{8}{9.6}{\rmdefault}{\mddefault}{\updefault}$1$}}}}
\put(5828,1937){\makebox(0,0)[lb]{\smash{{\SetFigFont{7}{8.4}{\familydefault}{\mddefault}{\updefault}$n-2$}}}}
\put(3010,2120){\makebox(0,0)[lb]{\smash{{\SetFigFont{8}{9.6}{\familydefault}{\mddefault}{\updefault}$a$}}}}
\put(2514,1925){\makebox(0,0)[lb]{\smash{{\SetFigFont{8}{9.6}{\rmdefault}{\mddefault}{\updefault}$2$}}}}
\put(3602,1941){\makebox(0,0)[lb]{\smash{{\SetFigFont{8}{9.6}{\rmdefault}{\mddefault}{\updefault}$3$}}}}
\put(6862,1937){\makebox(0,0)[lb]{\smash{{\SetFigFont{7}{8.4}{\familydefault}{\mddefault}{\updefault}$n-1$}}}}
\put(858,2105){\makebox(0,0)[lb]{\smash{{\SetFigFont{8}{9.6}{\familydefault}{\mddefault}{\updefault}$e$}}}}
\put(4083,2120){\makebox(0,0)[lb]{\smash{{\SetFigFont{8}{9.6}{\familydefault}{\mddefault}{\updefault}$a$}}}}
\put(5448,2134){\makebox(0,0)[lb]{\smash{{\SetFigFont{8}{9.6}{\familydefault}{\mddefault}{\updefault}$a$}}}}
\put(6505,2127){\makebox(0,0)[lb]{\smash{{\SetFigFont{8}{9.6}{\familydefault}{\mddefault}{\updefault}$a$}}}}
\put(2590,1423){\makebox(0,0)[lb]{\smash{{\SetFigFont{8}{9.6}{\familydefault}{\mddefault}{\updefault}$e$}}}}
\put(1973,1619){\makebox(0,0)[lb]{\smash{{\SetFigFont{8}{9.6}{\familydefault}{\mddefault}{\updefault}$e$}}}}
\put(3633,945){\makebox(0,0)[lb]{\smash{{\SetFigFont{8}{9.6}{\familydefault}{\mddefault}{\updefault}$e$}}}}
\put(3485,531){\makebox(0,0)[lb]{\smash{{\SetFigFont{8}{9.6}{\familydefault}{\mddefault}{\updefault}$a,c,e$}}}}
\put(4698,1941){\makebox(0,0)[lb]{\smash{{\SetFigFont{8}{9.6}{\familydefault}{\mddefault}{\updefault}$\cdots$}}}}
\put(1944,2121){\makebox(0,0)[lb]{\smash{{\SetFigFont{8}{9.6}{\familydefault}{\mddefault}{\updefault}$a$}}}}
\put(294,1918){\makebox(0,0)[lb]{\smash{{\SetFigFont{8}{9.6}{\rmdefault}{\mddefault}{\updefault}$0$}}}}
\put(3671,90){\makebox(0,0)[lb]{\smash{{\SetFigFont{8}{9.6}{\familydefault}{\mddefault}{\updefault}$d$}}}}
\put(98,2787){\makebox(0,0)[lb]{\smash{{\SetFigFont{8}{9.6}{\familydefault}{\mddefault}{\updefault}$a,c,d$}}}}
\put(1269,2743){\makebox(0,0)[lb]{\smash{{\SetFigFont{8}{9.6}{\familydefault}{\mddefault}{\updefault}$c,d,e$}}}}
\put(2423,2727){\makebox(0,0)[lb]{\smash{{\SetFigFont{8}{9.6}{\familydefault}{\mddefault}{\updefault}$c,d$}}}}
\put(3521,2720){\makebox(0,0)[lb]{\smash{{\SetFigFont{8}{9.6}{\familydefault}{\mddefault}{\updefault}$c,d$}}}}
\put(2565.000,2448.333){\arc{333.333}{2.2143}{7.2105}}
\blacken\path(2702.638,2411.417)(2665.000,2315.000)(2743.107,2382.913)(2702.638,2411.417)
\end{picture}
}
\end{center}
\caption{Nondeterministic automaton of the reverse of a left ideal.} 
\label{fig:LeftReverse}
\end{figure}
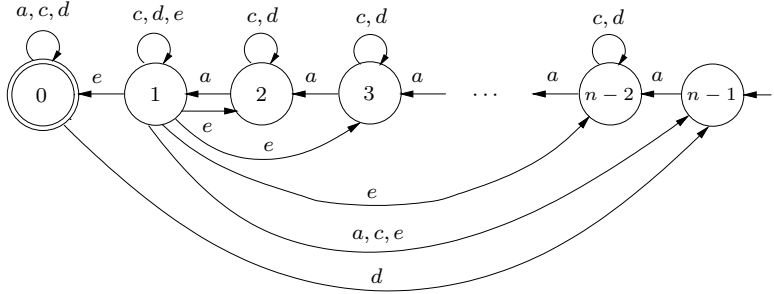

\section{Two-Sided Ideals and Factor-Closed Languages}
\label{sec:2Ideals}
We now consider two-sided ideals and factor-closed languages. We provide support for the following conjecture:
\medskip

\noin
{\bf Conjecture 2 (Two-Sided Ideals and Factor-Closed Languages).}
\label{con:leftideals}
{\it If $L$ is a two-sided ideal or a factor-closed language with quotient complexity $\kappa(L)=n\ge 2$, then it has syntactic complexity $\sig(L)\le n^{n-2} + (n-2) 2^{n-2} +1$.}
\medskip

We show in this section that this complexity can be reached.
Since every factor-closed language other than $\Sigma^*$ is the complement of a two-sided ideal, and complementation preserves syntactic complexity, it suffices to consider only two-sided ideals.

For $n=1$, the bound of the conjecture does not apply. The only two-sided ideal is $L=\Sig^*$, and it has $\sig(L)=1$.

For $n=2$ and $\Sig=\{a,b\}$, the only two-sided ideal is $L=\Sig^*a\Sig^*$, and it has $\sig(L)=2$, which is the bound of the conjecture.

For $n=3$ and $\Sig=\{a,b,c\}$, the automaton with inputs $a:[1,2,2]$, $b:[0,0,2]$, and $c:[0,1,2]$ has $\sig(L)=6$, which is the bound of the conjecture.

\begin{definition}
\label{def:Alin}
Let $n\ge 4$, and let $\cA_n$ be the  automaton 
$$\cA_n =(\{0,\ldots,n-1\},\{a,b,c,d,e,f\},\delta, 0,\{n-1\}),$$
where $a=(1,2,\ldots,n-2)$,
$b=(1,2)$,
$c={n-2\choose 1}$,
$d={n-2\choose 0}$,
for $i=0,\ldots,n-2$, $\delta(i,e)=1$ and $\delta(n-1,e)=n-1$,
and $f={1\choose n-1}$.
The state graph of $\cA_n$ is shown in Fig.~\ref{fig:TwoSided}.
For $n=4$, inputs $a$ and $b$ coincide.
\end{definition}

\begin{figure}[hbt]
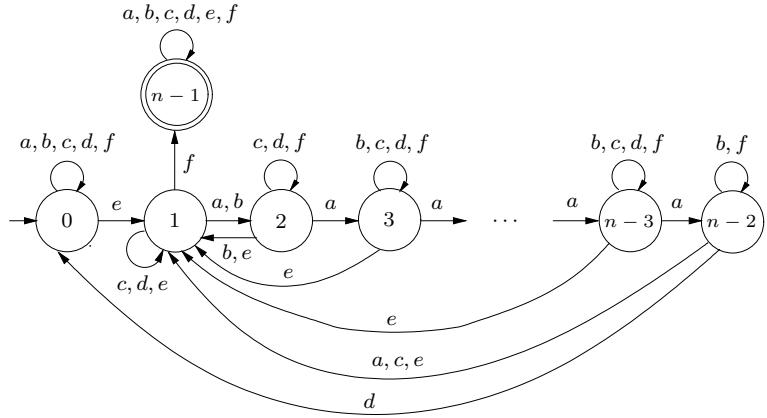

\begin{center}
\input TwoSided.eepic
\end{center}
\caption{Automaton $\cA_n$ of a two-sided ideal with $n^{n-2} + (n-2) 2^{n-2} +1$ transformations.} 
\label{fig:TwoSided}
\end{figure}

\begin{theorem}[Two-Sided Ideals and Factor-Closed Languages]
\label{thm:TwoSided}
Automaton $\cA_n$ of Fig.~\ref{fig:TwoSided} is minimal and the language $L=L(\cA_n)$ accepted by $\cA_n$ is a two-sided ideal and has syntactic complexity $\sig(L)=n^{n-2} + (n-2) 2^{n-2} +1$.
\end{theorem}
\begin{proof}
For $i=1,\ldots,n-2$, state $i$ is the only non-final state that accepts $a^{n-1-i}f$; hence all these states are distinguishable. State 0 is distinguishable from these states, because it does not accept any words in $a^*f$. Hence $\cA_n$ is minimal.
The proof that $\cA_n$ is a left ideal is like that in Theorem~\ref{thm:LeftIdeal4}. Since $L_{ef}=\Sig^*$ is the only accepting quotient, $L$ is a right ideal. Hence it is two-sided.

Consider any transformation $t$ of the form
$$
t=\left( \begin{array}{cccccccc}
0 & 1 & 2 & 3 & \cdots & n-3 & n-2 & n-1 \\
0 & i_1 & i_2 & i_3 & \cdots & i_{n-3} & i_{n-2} & n-1
\end{array} \right ),
$$
where $i_k\in\{0,1,2,\ldots,n-2, n-1\}$ for $1\le k\le n-1$;
there are $n^{n-2}$ such transformations.
We have two cases:
\be
\item If $i_k\ne n-1$ for all $k$, $1\le k\le n-2$, then all the images of the first $n-2$ states are in the set $\{0,\ldots,n-2\}$.
By Theorem~\ref{thm:salomaa}, $t$ can be done by $\cA_n$.
\item If $i_h = n-1$ for some $h$, $1\le h\le n-2$, then there exists some $j$,
$1\le j\le n-2$ such that $i_k\ne j$ for all $k$, $1\le k\le n-2$.
\ee
Define $i'_k$ for all $1\le k\le n-2$ as follows:
\begin{equation*} 
i'_k = \left\{ 
\begin{array}{ccc} 
j,	& & \quad \text{if } i_k=n-1;\\ 
i_k,   &   & \quad \text{if }  i_k\ne n-1.\\ 
\end{array} \right. 
\end{equation*}
Let 
$$
s=\left( \begin{array}{cccccccc}
0 & 1 & 2 & 3 & \cdots & n-3 & n-2 & n-1 \\
0 & i'_1 & i'_2 & i'_3 & \cdots & i'_{n-3} & i'_{n-2} & n-1 
\end{array} \right ),$$
and
$r=(1,j)$.
By Theorem~\ref{thm:salomaa}, $s$ and $r$ can be performed by $\cA_n$.

Now consider $srfr$.
If $t$ maps $k$ to $n-1$, then $s$ maps $k$ to $j$,  $r$ maps $j$ to $1$, 
$f$ maps $1$ to $n-1$, and $r$ maps $n-1$ to $n-1$.
If $t$ maps $k$ to $1$, then $s$ maps $k$ to $1$, $r$ maps $1$ to $j$, $f$ maps $j$ to $j$, and $r$ maps $j$ to $1$.
Finally, if $t$ maps $k$ to an element other than 1 or $n-1$, then $srfr$ maps $k$ to the same element.
Hence we have $t=srfr$, and $t$ can be performed by $\cA_n$ as well.
\medskip

Refer to states in $\{1,\ldots,n-2\}$ as the \emph{middle} states.
Take any transformation $t$ that maps 0 to $k\in\{1,\ldots,n-2\}$, and 
any middle state to either state in $\{i,n-1\}$. There are $(n-2)2^{n-2}$ such transformations.
First consider any entry $i$ that is mapped to $n-1$ by $t$. We can map $i$ to $n-1$ without changing any other states. First, apply $a^{n-1-i}$'s to rotate all the middle states clockwise, so that $i$ is mapped to $1$,
then apply $f$ to map $i$ to $n-1$, and then $a^i$ to return all the states other than $n-1$ to their original positions.
This is repeated for all the states that are mapped to $n-1$ by $t$.
After this is done, apply  $e$ to replace all the middle states by 1, and apply $a^{i-1}$ to move 1 to $i$. Hence $t$ can be performed.

Finally, the constant transformation $Q\choose n-1$ is done by $ef$.

In summary, the syntactic complexity of the language accepted by $\cA_n$ is at least $n^{n-2}+(n-2)2^{n-2}+1$.

Note that $0$ is mapped to a middle state $1$ if and only if the input word contains an $e$.
But every word of the form $xe$ leaves the automaton in a state in $\{1,n-1\}$. Applying any other word can only result in a state in $\{i,n-1\}$, for some middle state $i$.  Hence no transformations other  than the ones we have considered can be done by $\cA_n$, and the syntactic complexity of the language accepted by $\cA_n$ is precisely $n^{n-2}+(n-2)2^{n-2}+1$.
\qed

\end{proof}

Table~\ref{tab:2ISummary} summarizes our results for two-sided ideals.
For $\Sig=\{a,b\}$, the  values are reached by the languages $\Sig^*a^{n-1}\Sig^*$ for $n\ge 2$.
For $n=4$, $|\Sig|=3$, the value 16 is reached by $\cA_4$ restricted to $\{a,e,f\}$.
For  $|\Sig|=4$, the value 23 is reached by $\cA_4$ restricted to $\{a,d,e,f\}$.
For $n=5$, $|\Sig|=3$, the value 47 is reached by $\cA_5$ restricted to $\{a,e,f\}$.
For  $|\Sig|=4$, the value 90 is reached by $\cA_5$ restricted to $\{a,d,e,f\}$.
For  $|\Sig|=5$, the value 90 is reached by $\cA_5$ restricted to $\{a,c,d,e,f\}$. 

\begin{table}[ht]
\caption{Syntactic complexities for two-sided ideals.}
\label{tab:2ISummary}
\begin{center}
$
\begin{array}{| c ||c|c| c| c|c|c|c|}    
\hline
\ \  \ \ &\ \ n=1 \ \ &\ \ n=2 \ \ &\ \ n=3 \  \ & \  \  n=4 \ \ 
&  \ n=5 \  &\ \ \ \ 
&\  n=n \  
  \\
\hline  \hline
  |\Sig|=1
& \bf 1&\bf 1  &\bf 2&\bf  3  	&\bf 4    &   \ldots 	& \bf n-1 \\
\hline
|\Sig|=2 
& - & \bf 2 & 5 &   11	& 19 &     \ldots&  \\
\hline
|\Sig|=3
& - & - & 6 &  16	& 47  &   \ldots  	& \\
\hline
|\Sig|=4
& - &-  & - & 23 	& 90 &      \ldots	& \\\hline
|\Sig|=5
& - &-  & -        & 25	& 147  &   \ldots	&  \\
\hline
|\Sig|=6
& - &-  & -        &  -	& 150 &   \ldots	& n^{n-2} +(n-2)2^{n-2}+1 \\
\hline
\end{array}
$
\end{center}
\end{table}

 Our previous two results about reversal apply here as well.
\begin{theorem}[Reverse of Two-Sided Ideal]
The reverse of the two-sided ideal accepted by automaton $\cA_n$ of Theorem~\ref{thm:TwoSided} restricted to $\{a,d,e,f\}$
has $2^{n-2}+1$ quotients, which is the maximum possible for a two-sided  ideal.
\end{theorem}
\begin{proof}
Consider the subset construction applied to the nondeterministic automaton of Fig.~\ref{fig:2Reverse}.
Let $P=Q\setminus \{0,n-1\}$. We will show that $Q$ and all sets of the form $\{n-1\}\cup S$, where $S\subseteq P$
are reachable.
First, $Q$ is reached by $fe$.
Also,  $\delta(\{n-1\},(fa)^{n-3}f)=\{n-1\}\cup P$. To remove $i$, $1\le i \le n-2$, from any set
$\{n-1\}\cup S$, apply $a^{i-1}d$; this also rotates the remaining states of $P$ to the left by $i-1$ positions. Then apply $a^{n-2-(i-1)}$ to return the remaining states to their  original positions.
Hence all sets of the form $\{n-1\}\cup S$ are reachable. One verifies that all the $2^{n-2}+1$ subsets are pairwise distinguishable.
\qed
\end{proof}

\begin{figure}[hbt]
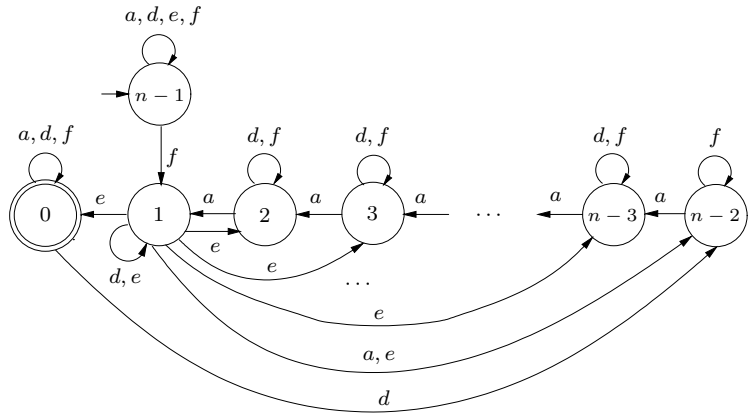

\begin{center}
\input 2Reverse.eepic
\end{center}
\caption{Nondeterministic automaton of the reverse of a two-sided ideal.} 
\label{fig:2Reverse}
\end{figure}

 \label{sec:conclusions}
Despite the fact that the Myhill congruence has left-right symmetry, there are significant differences between left and right ideals. 
The major open problem concerning ideals is to find a better upper bound for left ideals. 
Also, the relation between syntactic complexity and reversal deserves further study.
\providecommand{\noopsort}[1]{}

\end{document}